\documentclass[11pt,a4paper]{article}
\usepackage{amsmath,amsfonts,amsthm,bbm}
\usepackage[a4paper,top=3cm,bottom=2cm,left=3cm,right=3cm,marginparwidth=2cm]{geometry}
\usepackage[utf8]{inputenc}
\usepackage{color}
\newtheorem{proposition}{Proposition}
\theoremstyle{remark}
\newtheorem{remark}{Remark}

\usepackage{gensymb}
\usepackage{mathtools}
\usepackage{indentfirst}
\usepackage{listings}
\usepackage{soul}

\usepackage{booktabs}
\usepackage{lscape}
\usepackage{diagbox}
\usepackage{textcomp}
\usepackage{eurosym}

\usepackage[colorinlistoftodos,textwidth=4.5cm]{todonotes}

\title{Decision making with dynamic probabilistic forecasts\thanks{We thank Zorana Grbac for insightful discussions at an early stage of this project. Financial support from the Agence Nationale de Recherche
(project EcoREES ANR-19-CE05-0042) and from the
FIME Research Initiative is gratefully acknowledged. }}
\author{Peter Tankov\footnote{Corresponding author, email: \texttt{peter.tankov@ensae.fr}} \\ CREST -- ENSAE, Institut Polytechnique de Paris \and Laura Tinsi\\ EDF R\&D and CREST -- ENSAE, Institut Polytechnique de Paris}

\date{}

\begin{document}

\maketitle

\begin{abstract}

We consider a sequential decision making process, such as renewable energy trading or electrical production scheduling, whose outcome depends on the future realization of a random factor, such as a
meteorological variable. We assume that the decision maker disposes of a dynamically updated
probabilistic forecast (predictive distribution) of the random factor. We propose several stochastic models for the
evolution of the probabilistic forecast, and show how these models may
be calibrated from ensemble forecasts, commonly provided by weather centers. We then show how these
stochastic models can be used to determine optimal decision making
strategies depending on the forecast updates. Applications to wind
energy trading are given.
\end{abstract}

Key words: Probabilistic forecasting, ensemble forecasting, stochastic control, wind power trading

\section{Introduction}%\todo{Il serait bien de mettre dans l'introduction un graphique d'évolution d'une prévision probabiliste avec larger de l'intervalle de confiance qui varie avec le temps, pour motiver l'étude}

Consider a sequential decision-making process, such as renewable energy trading or electrical production scheduling, whose outcome depends on the realization
of a random factor, such as a meteorological variable. It
is often the case, that at each point in time, the
decision maker disposes of an imperfect probabilistic forecast of the random factor
(such as, a confidence interval or a set of quantiles), and that this
forecast is periodically, or continuously, updated. The goal of the
decision maker is to optimally update her strategy according to the
available information, to maximize a specific gain functional.  
To
solve this problem in the framework of stochastic
control, one needs to describe the dynamics of the predictive
distribution with a stochastic model. Such a model determines the evolution of the predictive
distribution and the relationship of the forecasts to the
realization of the unknown
random factor; in other words, the model describes the  evolution of the forecast
error as new information becomes available. 

In the literature, stochastic decision update rules based on point forecasts have been proposed \cite{tan2018optimal, collet2017optimal, aid2016optimal,skajaa2015intraday}, however probabilistic forecasts contain more dynamic information than point forecasts, as the expected forecast uncertainty can also vary dynamically. Figure \ref{forecast_evolution.fig} shows the evolution of the probabilistic forecast of power production of a wind plant in France as function of time, for a fixed production time. It is clear that not only the average production varies with time, but also the width of the confidence interval changes: it does not always decrease with time and may not be fully correlated with the expected production level. This information reflects the varying forecast uncertainty and is not contained in the point forecast, but may be important for decision making. For example, a wind producer facing severe penalties in case of lack of production, or a network operator whose goal is to avoid shortages at all costs may need to purchase energy in the intraday market to hedge the risk when forecast uncertainty increases, even if the predicted average production remains the same. This paper develops models of the dynamic evolution of probabilistic forecasts, allowing to take into account precisely this type of uncertainty in dynamic decision making. 

\begin{figure} 
  \centerline{\includegraphics[width=0.7\textwidth]{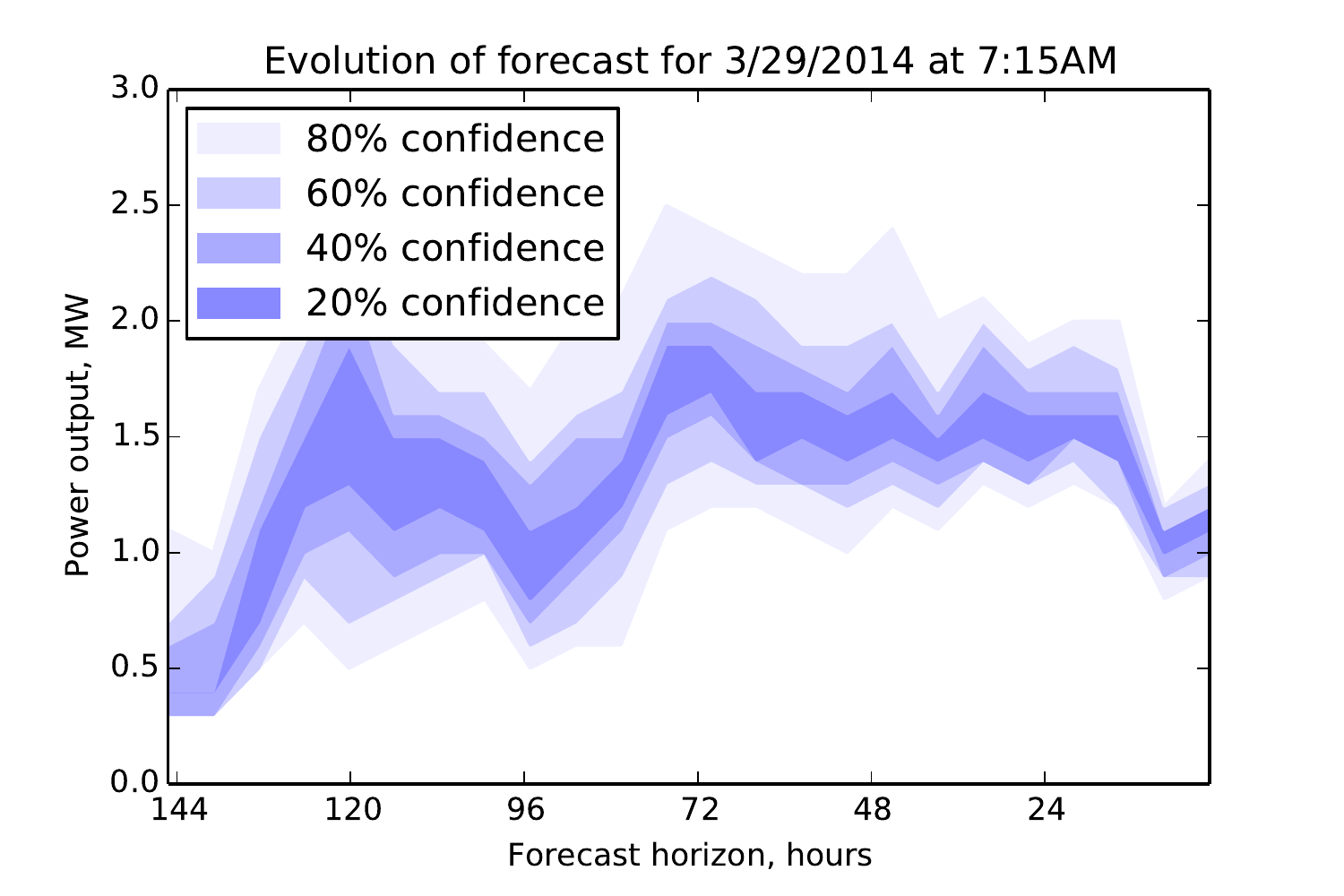}}
  \caption{Evolution of the probabilistic forecast of power production of a wind plant as function of time, for a fixed production time.}
  \label{forecast_evolution.fig} 
\end{figure}  

In mathematical terms, let $(\Omega,\mathcal F,\mathbb F, \mathbb P)$
be a filtered probability space and assume that $\mathbb F$ models the
filtration of the decision maker. Fix a time horizon $T$, and let $X$
be a real-valued $\mathcal F_T$-measurable random variable. We make a
standing assumption that $\mathcal F_0$ is a trivial $\sigma$-field. Let
$\mu_t$ denote the regular conditional distribution of $X$ given $\mathcal
F_t$. We call $\mu_t$ the probabilistic forecast of $X$ at time $t$. See \cite{gneiting2007probabilistic} for the description of the mathematical framework of probabilistic forecasting and methods of forecast evaluation. The goal of this paper is to
\begin{itemize}
\item[a.] Formulate the conditions that the dynamics of $\mu_t$ must satisfy and propose several tractable finite-dimensional models for this dynamics
 in the diffusion framework,
\item[b.] Show how these models may be calibrated  with real meteorologic data,  in the case where $\mu_t$ models the forecast of a meteorological variable. 
\item[c.] Provide an example of using the methodology to solve stochastic control problems arising in the context of wind energy trading. 
\end{itemize}

The
full predictive distribution is an infinite-dimensional object, but
the actual available information is always low-dimensional; for this
reason we aim to summarize the dynamics of the full
predictive distribution with a low number of factors, which are easy
to interpret and estimate from the data (such as the conditional mean and variance of
the predictive distribution). In addition, our objective of computing the
optimal strategies using the tools of stochastic control precludes the
use of high-dimensional specifications. 
More precisely, in this paper we consider parametric two-dimensional
specifications where the predictive distribution is a function of
two observable factors, say $m_t$
and $V_t$. Here $m_t$ represents the conditional expectation of
$X$ and $V_t$ some measure of the error, such as the conditional
variance. In our models, $m_t$ and $V_t$ have diffusion dynamics, and the predictive density $\mu_t$ corresponds to a distribution from some known class, such as Student t, normal inverse Gaussian, inverse Gaussian or log generalized hyperbolic, with parameters depending on $m_t$ and $V_t$.

%{ From the diffusion of $m_t$ and $V_t$, we may derive parametrized predictive densities. To this extent, some technical elements from the classical financial mathematics literature will be useful. Among them, Cox, Ingersoll and Ross \cite{cox1985theory}, and Dufresne \cite{dufresne1990distribution} offer good insights on Laplace transform for root processes and links between Brownian motions and gamma distributions respectively; while Barndorff-Nielsen \cite{barndorff1997processes}, and Aas and Haff \cite{aas2006generalized} develop convenient elements on the Generalized Hyperbolic Skew Student t distribution and the Normal Inverse Gaussian distribution.}\\

In practice, the forecast information received by the decision maker from a forecast provider may come, for example, in the form of a confidence interval around a point forecast, or in the form a set of quantiles of the predictive distribution. A particularly important case is that of ensemble forecasts. An ensemble forecast in meteorology is a set of several point forecasts aiming together to give an indication of the range of possible future states of the atmosphere. Members of the ensemble are obtained by running the forecasting model with perturbed initial conditions and / or parameters. An ensemble forecast is usually obtained with deterministic means, and therefore does not represent the best approximation of the predictive distribution of meteorological variables. In particular, ensemble forecasts are often uncalibrated (biased) and underdispersed compared to realizations \cite{gneiting2005calibrated}. However, techniques for statistical post-processing of ensemble forecast with the aim to improve calibration and sharpness have been developed in the literature. Two such techniques are ensemble model output statistics (EMOS) \cite{gneiting2005calibrated,thorarinsdottir2010probabilistic} and Bayesian model averaging (BMA) \cite{wilks2002smoothing,raftery2005using}. In \cite{gneiting2005calibrated}, the authors approximate the predictive density with a Gaussian distribution, whose parameters depend on the ensemble forecasts and are chosen to optimize calibration and sharpness of the resulting probabilistic forecast. In  \cite{wilks2002smoothing,raftery2005using} the predictive density is reprensented by a mixture of normal distributions, whose weights are computed from the ensemble members. To account for positive random variables such as wind speed, EMOS with log-normal distributions has been used in \cite{baran2015log} and BMA with  truncated normal components has been employed in \cite{baran2014probabilistic}. Other approaches to statistical post-processing of wind speed forecasts involve generalized extreme value distribution \cite{lerch2013comparison} and weighted mixtures of log-normal and truncated normal distributions \cite{baran2016mixture}.

In these papers, a single forecast horizon is fixed, and the calibration procedure uses a series of ensemble forecasts, obtained at different days of the training period for the fixed forecast horizon. At any given time, the calibrated method allows to compute the probabilistic forecast for this fixed horizon from the ensemble forecast, but no information about the evolution of the probabilistic forecast is available.

Our approach to calibrate the models presented in this paper is inspired by EMOS and also based on ensemble forecasts. However, we use more general predictive densities, potentially allowing for better calibration. More importantly, we do not fix a single forecast horizon, but model the dynamics of the predictive distribution for a given quantity at a given date, as time goes on and forecast horizon decreases. 
As a result, our calibrated model provides two types of information. First, as in statistical postprocessing methods, a predictive density in tractable form can be computed from an ensemble forecast. Secondly, the dynamics of this predictable distribution is given, in the form of a two-dimensional stochastic differential equation characterizing the evolution of the pair $(m_t,V_t)$, the conditional mean of the predictive distribution and a measure of the error. This dynamics can be exploited in the decision making process, to make strategy updates based not only on the conditional mean of the variable of interest, but also on the evolution of our knowledge of the uncertainty around the mean.

%In this article, probabilistic forecasts are approached in a different way: we don't make a% straight assumption on the conditional distribution of the $\mathcal{F}_T$-measurable random variable $X$, but rather on the dynamic of the two observable factors,  $m_t$ its conditional mean and $V_t$ its conditional variance, to derive the  distribution. The dynamics of $m_t$ and $V_t$ Thus, the purpose of the ensemble forecasts is three folds in this paper: first, calibrate the predictive density obtained from the stochastic model, then use it to deduce the left parameters of the diffusions of $m_t$ and $V_t$ not calibrated yet, third, when the decision maker has to act on the market, the shape of the calibrated processes as well as the incoming information on the empirical mean and variance of the ensemble forecasts will be used to determine uncertainty around the forecasts and elaborate the optimal strategy. 

Stochastic differential equations (SDE) have been used to model the dynamics of probabilistic forecasts by several authors, see e.g., \cite{iversen2016short,bensoussan2016cox} in the context of wind speed, or \cite{badosa2017day} in the context of solar energy forecasting. In these approaches, the forecasted quantity (e.g., the wind speed) is modeled directly by a stochastic differential equation, from which the probabilistic forecasts at any horizons, as well as their dynamics, can be deduced. However, the predictive distributions are typically not in tractable form (e.g., in \cite{iversen2016short} they are approximated by Monte Carlo), and the dynamics of the forecasting error is hard-coded into the equation and cannot be calibrated independently from ensemble forecasts, in other words, the variance of the forecast is not stochastic. This makes it impossible to use information on forecast uncertainty in strategy updates. 

Our approach provides a dynamic SDE-based model for forecast dynamics, tractable predictive distribution and possibility of model calibration with ensemble forecasts, in a sense taking the best of both worlds to obtain a coherent and realistic model. Moreover, the results are exploited in a stochastic control problem to integrate the additional information provided by the probabilistic forecasts in the decision process.\\

Using probabilistic forecasts for decision making in wind energy trading and electricity scheduling has been studied e.g., in \cite{pinson2007trading,pinson2013wind,zugno2013trading}. These references suggest a static approach, where a probabilistic forecast of a quantity of interest is used to make the decision on e.g., the quantity of energy to sell in the day-ahead market. By contrast, our dynamic approach allows to continuously, or regularly, update the decision based on the evolution of the forecast and information about its uncertainty.

 We illustrate our methodological contribution with an application to a wind power trading problem. In this problem, a wind power producer, who disposes of a dynamically updated probabilistic forecast of the wind speed, takes positions in the intraday electricity market to maximize the utility of terminal wealth. This setting gives rise to a three-dimensional stochastic control problem, which is solved using the dynamic programming principle. For the numerical solution we use the Least Squares Monte Carlo method implemented in the open-source library StOpt (see \cite{gevret2018stochastic}) and based on the methods of Bouchard and Warin \cite{bouchard2012monte} and Belomestny et al. \cite{belomestny2010regression} generalizing the seminal approach of Longstaff and Schwartz \cite{longstaff2001valuing} and Tsitsiklis and Van Roy \cite{tsitsiklis1999optimal}. To assess the value of taking into account the dynamics of probabilistic forecasts, we compare the gains of an agent using our approach with the potential gains of another agent who uses only the point forecasts and show that our method leads to a 5\% revenue increase in the simulation examples.

The paper is structured as follows. In section \ref{model} we describe several parametric models for the dynamics of probabilistic forecasts. In section \ref{calibration}, we develop a procedure inspired by EMOS to calibrate the models of section \ref{model} for different lead times and show that  our models have good prediction results and that EMOS  increases accuracy of prediction compared with raw ensembles, as expressed with Continuous Ranked Probability Score. In section \ref{controlpb}, we present an application of our methodology to wind power trading.

\section{Modeling probabilistic forecasts}\label{model}
%Two-dimensional specifications correspond to the situation when, in
%addition to the point forecast, exogeneous information about its
%precision, such as a confidence interval, standard deviation, or a set
%of quantiles, is available. Since it seems difficult to construct a tractable %two-dimensional
%specification for an arbitrarily given terminal distribution, we
%propose below several specific parametric examples, which should be
%sufficient for problems encountered in practice. 
As mentioned in the introduction, given a flow of information described by the filtration $(\mathcal F_t)_{0\leq t\leq T}$, a probabilistic forecast of an $\mathcal F_T$-measurable random variable $X \in \mathbb R^d$ is the conditional distribution $\mu_t$ of $X$ given $\mathcal F_t$. A dynamic model for a probabilistic forecast is then a flow of probability measures $(\mu_t)_{0\leq t\leq T}$, which can be identified with a flow of conditional distributions of some $\mathcal F_T$-measurable random variable. This imposes strong constraints on the dynamics of $(\mu_t)$, in particular, all moments of $\mu_t$, when they exist, must be $(\mathcal F_t)$-martingales. 
A $d$-dimensional Markov specification of forecast dynamics is a Markov process $(X_t)_{0\leq t\leq T}\in \mathbb R^d$ such that, at every $t\in [0,T]$, $\mu_t = \mu(t,X_t)$, where $\mu:[0,T]\to \mathbb R^d \to \mathcal P(\mathbb R)$ is a deterministic mapping, where $\mathcal P(\mathbb R)$ is the set of probability measures on $\mathbb R$. 

In this section, we develop several two-dimensional Markov specifications for forecast dynamics, which correspond to well-known tractable predictive distributions. %\todo{voir avec peter l'ordre des abscisses}

\subsection{Forecast of a real-valued quantity}
In this section we propose two tractable models for the dynamics of
probabilistic forecast of a real-valued quantity, such as the
temperature. The models are based on the time-changed Brownian
motion. In the first paragraph, the predictive distribution at all times is the
Student t distribution (with power law tails), and in the second
paragraph, the predictive distribution is the normal inverse Gaussian
distribution (with exponentially decaying tails). 
\paragraph{Student t predictive distribution}
Let $\rho$ be a positive deterministic function,
continuous on $(0,\infty)$, 
with $\int_0^t \rho^2 (t) dt = +\infty$ for all $t>0$ (this function is singular at zero), let $W$ and $W'$ be
independent standard Brownian motions, whose filtration will be
denoted by $\mathbb F$, let $b>0$ and consider the following pair of
stochastic differential equations, defined for $t\in [0,T)$:
\begin{align}
\frac{dV_t}{V_t} &= -\rho^2(T-t) dt + b \rho(T-t) dW_t\label{ln1}\\
 dm_t &=\sqrt{V_t} \rho(T-t)dW'_t,\label{ln2}
\end{align}
\begin{proposition}\label{student.prop}
The equation (\ref{ln1}--\ref{ln2}) admits a strong solution $(m,V)$ on $[0,T)$. The limit $m_T = \lim_{t\to T} m_t$ exists in the almost sure sense, and for every $t\in[0,T)$, the
conditional distribution of $m_T$ given $\mathcal F_t$ is the Student
t distribution with $2\nu$ degrees of freedom, where $\nu = 1+
2/b^2$:
\begin{align*}
\frac{d}{dx}\mathbb P[m_T\in dx|\mathcal F_t]&= \frac{d}{dx}\mathbb
P[m_T\in dx|m_t, V_t] \\ &= 
\frac{\Gamma(\nu + \frac{1}{2})}{\Gamma(\nu)}\frac{b}{2 \sqrt{\pi V_t
    }} \left\{1+\frac{(x-m_t)^2 b^2}{4 V_t}\right\}^{-\nu-\frac{1}{2}}.
\end{align*}
In addition,
$$
m_t = \mathbb E[m_T|\mathcal F_t] \quad \text{and}\quad V_t =
\mathrm{Var}[m_T|\mathcal F_t]. 
$$
\end{proposition}

\begin{remark}
  In this model, and in the other models of this section, the predictive distribution is parameterized by two (stochastic) variable parameters, $m_t$ and $V_t$ which typically determine the location and scale of the distribution and may change as time passes and the forecast horizon draws near, and one fixed parameter $b$ (typically, the shape parameter), which remains constant throughout the lifetime of the forecast for a fixed date. In addition, the deterministic time-varying parameter $\rho$ does not affect the predictive distribution, but affects the dynamics of the variables  $m_t$ and $V_t$. 
\end{remark}  
\begin{remark}\label{timechange.rem}
The pair $(V_t,m_t)_{0\leq t< T}$ can alternatively be written as  $V_t = \widetilde V_{\theta_t}$ and $ m_t = \widetilde m_{\theta_t}$ with 
$$
\theta_t = \int_0^t \rho^2(T-s) ds,\quad 0\leq t< T,
$$
and
\begin{align*}
\widetilde V_t &= V_0 e^{-(1+\frac{b^2}{2})t + b \widetilde W_t}\\
    \widetilde m_t& = m_0 + \int_0^t \sqrt{\widetilde V_s} d\widetilde W'_s
\end{align*}
on $[0,\infty)$, with $(\widetilde W,\widetilde W')$ a standard $2$-dimensional Brownian motion. 
\end{remark}
\begin{proof}
Fixing $t<T$, as in the above remark, for $s\geq t$, we can write
$$
V_s = V_t\exp\left(-\int_t^s(1 + b^2/2) \rho^2(T-u)du + \int_t^s
  b \rho(T-u)
  dW_u \right) = V_t\overline V^{(-1-b^2/2,b)}_{\int_t^s \rho(T-u)^2 du}
$$
where $\overline V^{(\mu,b)}_t = e^{\mu t + b B'_t} $
for a different Brownian motion $B'$. In addition,
\begin{align*}
\int_t^T \rho^2(T-s) V_s dt &= V_t \int_t^T \rho^2(T-s) \overline V^{(-1-b^2/2,b)}_{\int_t^s
  \rho(T-u)^2 du} ds = V_t\int_0^{\infty} \overline V^{(-1-b^2/2,b)}_s ds \\&= \frac{4V_t}{b^2}\int_0^{\infty} \overline V^{-(4/b^2-2,2)}_s ds . 
\end{align*}

From \cite[Proposition 4.4.4]{dufresne1990distribution}, 
$$
\int_0^\infty\overline V_t^{(-4/b^2-2,2)} dt  \stackrel{d}{=}
(2\gamma_{\nu})^{-1},\quad \nu = 2/b^2 + 1,
$$
where $\gamma_\nu$ denotes a gamma random variable with parameter
$\nu$. 

%Assume that the Brownian motions $W$ and $W'$ are independent. 
Then, for $t<s<T$, 
$$
m_s =m_t + B\left(\int_t^s  V_r \rho^2(T-r)dr\right),
$$
for a different Brownian motion $B$. 
Therefore,
$$
m_T:=\lim_{s\to T}m_s=m_t + B\left(\int_t^T  V_r \rho^2(T-r)dr\right)=
B\left(\frac{2 V_t}{b^2 \gamma_\nu}\right)
$$
and finally
\begin{align*}
\frac{d}{dx}\mathbb P[m_T-m_t\in dx|\mathcal F_t] & =
\mathbb E\left[\frac{b\sqrt{\gamma_\nu}}{2\sqrt{\pi V_t
    }}e^{-\frac{x^2 b^2 \gamma_\nu}{4 V_t}}\right]\\
& = \frac{b}{2\Gamma(\nu) \sqrt{\pi V_t
    }} \int_0^\infty e^{-z-\frac{x^2 b^2 z}{4 V_t}} z^{\nu-\frac{1}{2}} dz \\
& = \frac{\Gamma(\nu + \frac{1}{2})}{\Gamma(\nu)}\frac{b}{2 \sqrt{\pi V_t
    }} \left\{1+\frac{x^2 b^2}{4 V_t}\right\}^{-\nu-\frac{1}{2}},
\end{align*}
which means that conditionnally on $\mathcal F_t$, $m_T-m_t$ follows the centered Student t distribution with $2\nu = 2+
4/b^2$
degrees of freedom. The expressions for mean and variance are obtained
from standard formulas for the Student t distribution. 
\end{proof}

Figure \ref{studentci} illustrates the dynamics of the predictive distribution in the model (\ref{ln1}--\ref{ln2}). We see that the confidence interval has a nontrivial behavior, it shrinks around the middle of the graph as the conditional variance goes down before increasing in size when the conditional variance goes up and shrinking to zero again at the very end. 
%\todo{Désolé mais il faut encore modifier le graphique. Dans le caption il faut mettre juste 'evolution of predictive density quantiles', et sur l'axe des x ce n'est pas forecast horizon mais simplement le temps}
\begin{figure}[ht]
      \centering
      \caption{Confidence intervals for the Student predictive density, and the trajectories of the predictive distribution mean $m_t$ and variance $V_t$. }
      \includegraphics[width = 0.7\textwidth]{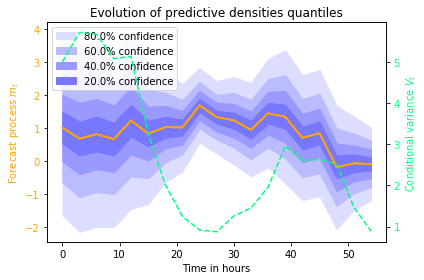}
      \label{studentci}
\end{figure}
\paragraph{Normal inverse Gaussian predictive distribution}
Using the notation of the
previous paragraph, consider the following pair of stochastic
differential equations. 
\begin{align}
dm_t &= \sqrt{{V}_t} \rho(T-t) dW_t\label{nig1}\\
d{V}_t &= - {V}_t \rho^2(T-t)dt + \sqrt{{V}_t} b
\rho(T-t) dW'_t.\label{nig2}
\end{align}
Here ${V}$ is a time-changed square-root process which hits zero
in finite time almost surely. We assume that this process remains at zero after the
first hitting time. A sample evolution of $(m,V)$ and the dynamics of the associated predictive distribution is shown in Figure \ref{nigci}. 
As in Remark \ref{timechange.rem}, we can express $m_t = \widetilde m_{\theta_t}$ and ${V}_t = \widetilde {V}_{\theta_t}$, where the processes $\widetilde m$ and $\widetilde V$ have time-homogeneous
dynamics (with different Brownian motions). 
\begin{align}
{d\widetilde m_t} &= \sqrt{\widetilde{V}_t}  dW_t\label{nigtc1}\\
d\widetilde{V}_t &= - \widetilde{V}_t dt + \sqrt{\widetilde{V}_t} b
dW'_t,\label{nigtc2}
\end{align}
\begin{proposition}\label{nig.prop}
The equation (\ref{pos1}--\ref{pos2}) admits a strong solution $(m,V)$ on $[0,T)$. The limit $m_T = \lim_{t\to T} m_t$ exists in the almost sure sense, and for every $t\in[0,T)$, the
conditional distribution of $ {m_T}$ given $\mathcal F_t$ is the symmetric
normal inverse Gaussian distribution on $\mathbb R$ with density 
\begin{equation}\label{nig.dist}
    p(x) = \frac{\frac{{V}_t}{b^2} K_1\left(\frac{1}{b} \sqrt{({V}_t/b)^2 +
      (x-m_t)^2}\right)}{\pi \sqrt{({V}_t/b)^2+(x-m_t)^2}}
  e^{{V}_t/b^2 }
\end{equation}
where $K$ is the modified Bessel
function of the third kind. Moreover, the conditional
mean and variance of $m_T$ are given by
$$
m_t = \mathbb E[m_T|\mathcal F_t] \quad \text{and}\quad \mathrm{Var}\,[m_T |\mathcal F_t] = {V}_t.  
$$

\end{proposition}

\begin{proof}
For the existence of the strong solution to \eqref{nigtc1}--\eqref{nigtc2}, see \cite[Section 6.3.1]{jeanblanc2009mathematical}
From Remark \ref{timechange.rem} it follows that  for $s\geq \theta_t$, %\todo[inline]{Je préfère $S_t$ pour faire apparaitre explicitement la dépendance par rapport à t. On peut changer la notation et mettre par exemple $s_t$ ou $\theta_t$ pour éviter la confustion dans la partie numérique. Je te laisse décider.}
$$
\tilde m_s -m_t =\int_{\theta_t}^s \sqrt{\widetilde{V}_u}
  dW_u = \widetilde W_{\int_{\theta_t}^s \widetilde {V}_u du}
$$
for a different Brownian motion $\widetilde W$. 
In particular 
$$
m_T = \widetilde m_\infty = m_t + 
\widetilde W_{\int_{\theta_t}^\infty \widetilde {V}_s ds}.
$$ 
$\tilde {V}$ is a square root process with zero long-term
mean. The Laplace transform of the integrated square root process is
known \cite[Proposition 6.3.4.1]{jeanblanc2009mathematical}:
$$
\mathbb E\left[\exp\left(-u\int_{\theta_t}^s \widetilde {V}_u
    du\right)\Big|\widetilde {V}_{\theta_t}\right]
= \exp\left(-\frac{2 \widetilde {V}_{\theta_t} u}{1 + \gamma \coth
    \frac{\gamma (s-\theta_t)}{2}}\right),
$$
%\todo[inline]{J'ai bien vérifié dans Jeanblanc et Yor, la formule est juste}
where $\gamma = \sqrt{1 + 2u b^2}$. Integrating up to
infinity, we then find:
$$
\mathbb E\left[\exp\left(-u\int_{\theta_t}^\infty \widetilde {V}_s
    ds\right)\Big| \widetilde {V}_{\theta_t}\right]
= \exp\left(-\frac{2 \widetilde {V}_{\theta_t} u}{1 + \sqrt{1 + 2u b^2} }\right).
$$

This allows us to compute the Fourier transform of the conditional
distribution of $m_T$: 
\begin{align*}
\mathbb E[e^{iu ({m_T}-{m_t})} |\mathcal F_t] &= \mathbb E\left[\exp\left(iu
    \widetilde W_{\int_{\theta_t}^\infty \widetilde {V}_s ds}\right)|\widetilde {V}_{\theta_t}\right]\\
& = \mathbb E\left[\exp\left(-\frac{u^2}{2}\int_{\theta_t}^\infty \widetilde
    {V}_s ds\right)\Big| \widetilde {V}_{\theta_t}\right]\\
& = \exp\left(-\frac{\widetilde{V}_{\theta_t} u^2}{1 +
    \sqrt{1 + u^2b^2} }\right)\\
& = \exp\left(-\frac{{V}_{t} (
    \sqrt{1 + u^2b^2} -1)}{ b^2}\right)
\end{align*}
The characteristic function of the normal inverse Gaussian law with
parameters $\mu$, $\alpha$, $\beta$, $\delta$
\cite{barndorff1997processes} is given by
$$
e^{i\mu u + \delta(\sqrt{\alpha^2 - \beta^2} - \sqrt{\alpha^2 - (\beta+iu)^2})}.
$$
Hence, $m_T-m_t$ conditionnally on $\mathcal F_t$ follows the normal inverse Gaussian law
with parameters
$$
\mu = 0, \quad \delta = \frac{{V}_t}{b},\quad \beta =0,\quad \alpha  = \frac{1}{b}.
$$ 
The expressions of the conditional moments may be easily obtained from
the characteristic function. %\todo[inline]{est ce qu'il faudrait par faire apparaitre la dependance au temps ? (ie $\mu_t$ au lieu de $\mu$) - non, pas nécessaire, pour être cohérent avec la ligne précédente}
\end{proof}

\begin{figure}[ht]
      \centering
      \caption{Confidence intervals for the NIG predictive density, and the trajectories of the predictive distribution mean $m_t$ and variance $V_t$. }
      \includegraphics[width = 0.67\textwidth]{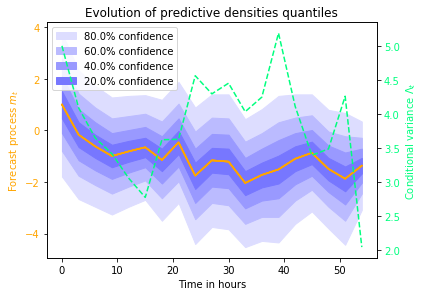}
      \label{nigci}
\end{figure}

\subsection{Forecast of a positive quantity}\label{sec:pos}
In this section we propose two models for the probabilistic forecast of a
positive quantity such as the wind speed. The models are obtained from
the ones of the previous section, replacing the Brownian motion with
the martingale geometric Brownian motion. 
\paragraph{Log-generalized hyperbolic predictive distribution}
Using the notation of the preceding section, consider the following pair of
stochastic differential equations.
\begin{align}
\frac{dV_t}{V_t} &= -\rho^2(T-t) dt + b \rho(T-t) dW_t\label{gh1}\\
 \frac{dm_t}{m_t} &=\sqrt{V_t} \rho(T-t)dW'_t.\label{gh2}
\end{align}
As in Remark \ref{timechange.rem}, we can write $m_t = \widetilde m_{\theta_t}$ and ${V}_t = \widetilde {V}_{\theta_t}$, where the processes $\widetilde m$ and $\widetilde V$ have time-homogeneous
dynamics (with different Brownian motions). 
\begin{align*}
\frac{d\widetilde m_t}{\widetilde m_t} &= \sqrt{\widetilde{V}_t}  dW_t\\
\frac{d\widetilde{V}_t}{\widetilde V_t} &= - dt +  b
dW'_t.
\end{align*}

\begin{proposition}
Let $(m,V)$ be a solution of (\ref{gh1}--\ref{gh2}). Then the
conditional distribution of $\log m_T$ given $\mathcal F_t$ is the
generalized hyperbolic distribution with density 
$$
p(x) = \frac{b e^{\frac{x-\mu}{2}}}{\Gamma(\nu) \sqrt{\pi V_t
    }}  \left(\frac{V_t}{b \sqrt{4V_t+ (x-\mu)^2 b^2}}\right)^{\nu + \frac{1}{2}} K_{\nu + \frac{1}{2}}\left(\sqrt{\frac{V_t}{b^2}+ \frac{(x-\mu)^2}{4}}\right)
$$
with $\mu = \log m_t$ and $\nu = 1+\frac{2}{b^2}$. In addition,
$$
\mathbb E[m_T|\mathcal F_t] = m_t.
$$
\end{proposition}
\begin{remark}
This distribution is a particular case of the generalized hyperbolic
distribution, known as generalized hyperbolic skew Student t
distribution \cite{aas2006generalized}. With this distribution, $m_T$
does not admit a second moment. 
\end{remark}
\begin{proof}
With the notation of the proof of Proposition \ref{student.prop}, we
now get
$$
\log\frac{m_T}{m_t} = W\left(\frac{2V_t}{b^2 \gamma_\nu}\right) -
\frac{V_t}{b^2 \gamma_\nu},
$$
and therefore
\begin{align*}
\frac{d}{dx}\mathbb P[\log(m_T/m_t)\in dx|\mathcal F_t] &= \frac{d}{dx}\mathbb
E[\mathbb P[\log(m_T/m_t)\in dx|\gamma_\mu,V_t]] \\ &=
\mathbb E\left[\frac{b\sqrt{\gamma_\mu}}{2\sqrt{\pi V_t
    }}e^{-\frac{(x-\frac{V_t}{b^2 \gamma_\mu})^2 b^2 \gamma_\mu}{4 V_t}}\right]\\
& = \frac{b e^{\frac{x}{2}}}{2\Gamma(\mu) \sqrt{\pi V_t
    }} \int_0^\infty e^{-z\left(1+ \frac{x^2 b^2 }{4V_t}\right) -
    \frac{V_t}{4b^2 z}} z^{\mu-\frac{1}{2}} dz \\
& = \frac{b e^{\frac{x}{2}}}{\Gamma(\mu) \sqrt{\pi V_t
    }}  \left(\frac{V_t}{b \sqrt{4V_t+ x^2 b^2}}\right)^{\mu + \frac{1}{2}} K_{\mu + \frac{1}{2}}\left(\sqrt{\frac{V_t}{b^2}+ \frac{x^2}{4}}\right).
\end{align*}
\end{proof}
\begin{figure}[ht]
      \centering
      \caption{Confidence intervals for the log-generalized hyperbolic predictive density and the trajectories of the processes $m$ and ${V}$ }
      \includegraphics[width = 0.67\textwidth]{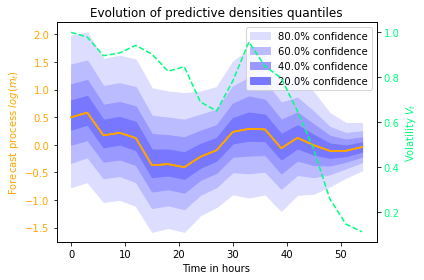}
      \label{ghci}
\end{figure}
\paragraph{Log-normal inverse Gaussian predictable distribution} 
Using the same notation as
above, consider the following pair of SDEs. 
\begin{align}
\frac{dm_t}{m_t} &= \sqrt{V_t} \rho(T-t) dW_t\label{pos1}\\
dV_t &= - V_t \rho^2(T-t)\left(1+\frac{b^2}{2}\right)dt + \sqrt{V_t} b
\rho(T-t) dW'_t.\label{pos2}
\end{align}
%\todo[inline]{C'est bizarre que tout d'un coup on change de notation et on utilise $V$ au lieu de $V$. Est-ce que tu peux remplacer partout $V$ par $V$? Y compris sur le graphique.}
The equivalent time-changed representation takes the form
\begin{align}
\frac{d\widetilde m_t}{\widetilde m_t} &= \sqrt{\widetilde V_t}  dW_t\label{pos1tc}\\
d\widetilde V_t &= - \widetilde V_t \left(1+\frac{b^2}{2}\right)dt + \sqrt{\widetilde V_t} b dW'_t.\label{pos2tc}
\end{align}
The proof of the following proposition is very similar to that of
Proposition \ref{nig.prop} and will therefore be omitted. 
%In the case of a positive quantity, it is more convenient to model the
%conditional mean $m_t$ and the logarithm of the ratio of the
%conditional second moment to the square of the conditional mean
%$V_t = \log\frac{M_t}{m^2_t}$. A tractable model is obtained by
%considering a square root model for the process $(V_t)$. 

\begin{proposition}
Let $(m,V)$ be a solution of (\ref{pos1}--\ref{pos2}). Then the
conditional distribution of $\log {m_T}$ given $\mathcal F_t$ is the
normal inverse Gaussian distribution on $\mathbb R$ with density 
\begin{equation}\label{logNig.dist}
   p(x) = \frac{\alpha \delta K_1\left(\alpha\sqrt{\delta^2 +
      (x-\mu)^2}\right)}{\pi \sqrt{\delta^2+(x-\mu)^2}}
  e^{\delta\gamma + \beta(x-\mu)} 
\end{equation}

where the parameters are given by 
$$
\mu = \log m_t, \quad \delta = \frac{V_t}{b},\quad \beta =
-\frac{1}{2},\quad \alpha  =  \sqrt{ (b^{-1}+\frac{b}{2})^2  + \frac{1}{4}},
$$
$\gamma = b^{-1} + \frac{b}{2}$, and $K$ is the modified Bessel
function of the third kind. Moreover, the first two conditional
moments of $m_T$ are given by
$$
m_t = \mathbb E[m_T|\mathcal F_t] \quad \text{and}\quad \mathbb
E[m_T^2 |\mathcal F_t] = m_t^2 e^{V_t}. 
$$

\end{proposition}
\begin{figure}[ht]
      \centering
      \caption{Confidence intervals for the log normal inverse Gaussian predictive density and the trajectories of the processes $m$ and $V$. }
      \includegraphics[width = 0.7\textwidth]{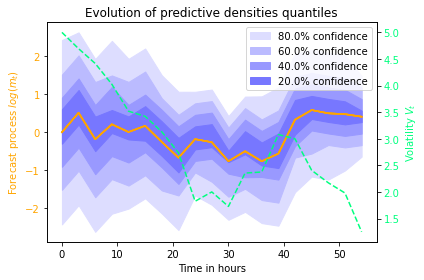}
      \label{lnigci}
\end{figure}

Note that for all processes presented in this section, the { process $V_t$}  does not necessarily decrease with time (see Figures \ref{studentci}-\ref{lnigci}). This reflects the fact that uncertainty over the quantity to forecast can vary over time and does not always decrease as we approach the realization date. {We also attract the reader's attention on the fact that while $V_t$  does represent the forecast uncertainty, it does not coincide with the conditional variance of the predictive density for positive quantities -- \textit{e.g.,} for $(m,V)$ solution of \eqref{pos1}--\eqref{pos2}, the variance is given by $\text{Var}[m_T|\mathcal{F}_t] = m^2_t(e^{V_t}-1)$.}

\section{Fitting forecast models to data}\label{calibration}
% In the paper, we aim to use the predictive distribution of a random factor regularly updated to take decisions. In order to capture the uncertainty related to the prediction we chose $2$-factors models depending, classically, on the empirical mean and variance of an ensemble of forecasts. In this section we would like to use the parametric predictive densities obtained in section \ref{model}, and the stochastic dynamic associated. 

In this section, we detail the procedure for calibrating our models for forecast dynamics from historical ensemble forecasts and the corresponding realizations.
To illustrate the forecasting of a real-valued quantity, we shall use the normal inverse Gaussian model defined by the equations (\ref{nig1}--\ref{nig2}) and the predictive density \eqref{nig.dist}, and apply it to ensemble forecasts of temperature. 
To ilustrate the forecasting of a positive quantity, we shall use the model defined by the equations (\ref{pos1}--\ref{pos2}) and the predictive density \eqref{logNig.dist}, and apply it to ensemble forecasts of the wind speed.

 \subsection{Presentation of the dataset}\label{datapres}
 
 The data is composed of meteorological ensemble forecasts from $K=273$ different locations {around Paris, France, plotted on the map in Figure \ref{locations},} recorded over January 2015.
 
\begin{figure}[ht]
      \centering
      \caption{Locations of the meteorological records }
      \includegraphics[width =0.5\textwidth]{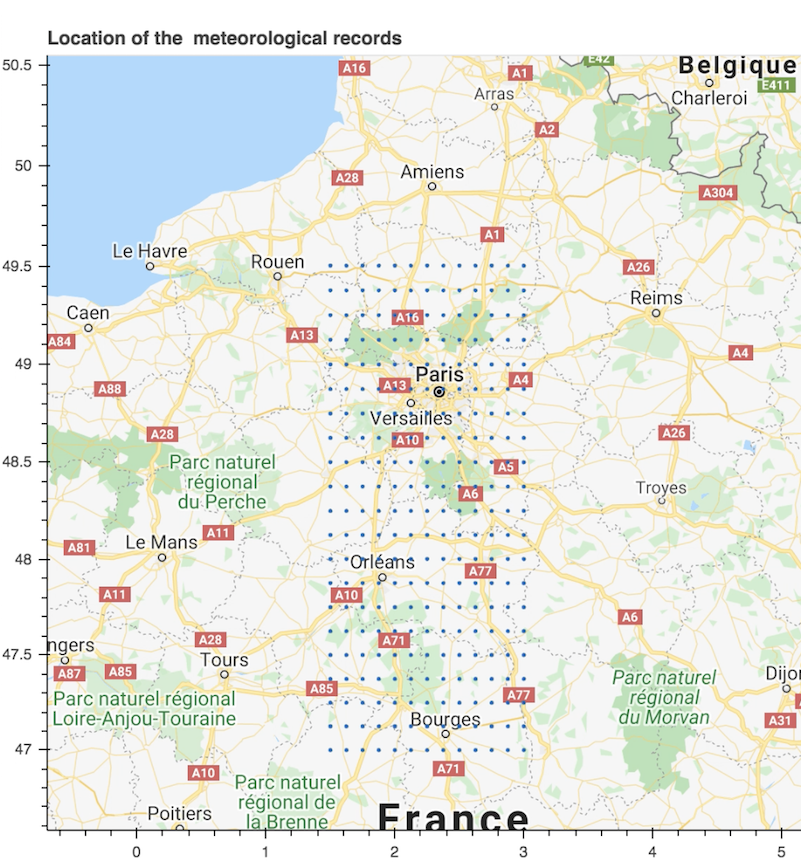}
      \label{locations}
\end{figure}

 A new forecast ensemble becomes available at 12PM (noon) and at 12AM (midnight) on each day. Each forecast ensemble consists of 50 members, and each member provides a prediction for all meteorological variables for lead times from 1h to 48h, with a step of {3h}. Since the forecasts are updated every 12 hours, in our study of forecasts dynamics, we use only the forecast horizons which are multiples of 12 hours, that is, $h\in \mathcal H = \{12,24,36,48\}$.

 As the locations are very close to each other we make the approximation that they form several ensemble forecasts of the same area.  { We will use all the ensemble forecasts to calibrate the model, making the approximation that, for each time horizon, the calibrated coefficient can be used to obtain the predictive densities in each of these locations.}

 Forecasts recorded on days from January 3 to January 21 constitute the training set and those recorded from January 22 to January 31 form the test set. 

 Note that the first two days of January are not used for the calibration because we do not dispose of the full forecast data for them. Our training set is thus composed of $T=38$ 12-hour periods.

  In this application we are interested in two variables: the temperature at 2 meter height  denoted by $\tau$ and the 10 meter wind speed denoted by $w$.  The wind speed is not directly available in the data and we compute it from the two components $w_x$ and $w_y$ through the usual formula $w= \sqrt{w_x^2+w_y^2}$ {for each member of the forecast ensemble and for each realization. }

  Since the maximum lead time for our forecast is 48 hours and new forecast becomes available every 12 hours, to study the dynamics of the forecast of a given realization recorded at 12 AM or 12 PM, we dispose of 4 data points with lead times 48h, 36h, 24h and 12h.  In Figure \ref{Fig1}, we plot four forecast ensembles for the wind speed, recorded for a specific location in our dataset at four consecutive forecast update times (Jan 1st 12PM, Jan 2nd 12AM, Jan 2nd 12 PM and Jan 3rd 12AM). The forecasts  for a fixed terminal time (Jan 3rd, 12 PM) are shown with the vertical bar in the four graphs. 
%\todo[inline]{Sur les graphiques de Fig. 7  je ne comprends pas la légende 'time to realization'. Est-ce que tu peux enlever la légende? Le plus simple est de le faire directement dans le pdf du graphique. Les titres des graphiques sont aussi incoherents: ça devrait être 'Forecast produced on Jan 1st at 12PM, Forecast produced on Jan 2nd at 12AM' etc. }
\begin{figure}[ht]
      \centering
      \caption{Forecast ensembles for the wind speed, recorded at four consecutive forecast update times (Jan 1st 12PM, Jan 2nd 12AM, Jan 2nd 12 PM and Jan 3rd 12AM). The forecasts  for a fixed terminal time (Jan 3rd, 12 PM) are shown with the vertical bar }
      \includegraphics[width = 0.45\textwidth]{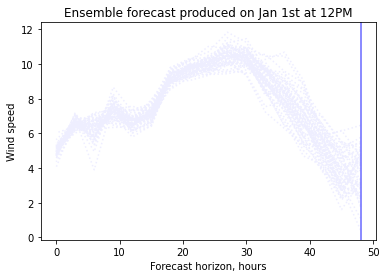}\label{Fig1}
      \includegraphics[width =0.45\textwidth]{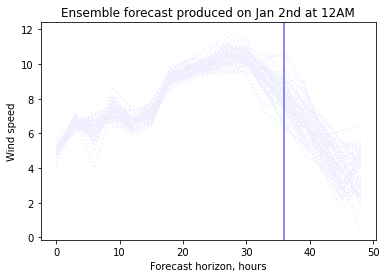}
    \includegraphics[width = 0.45\textwidth]{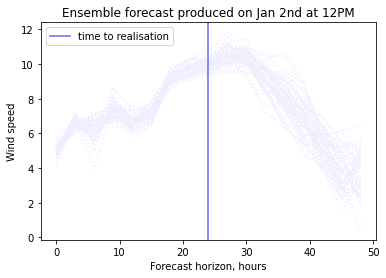}
      \includegraphics[width =0.45\textwidth]{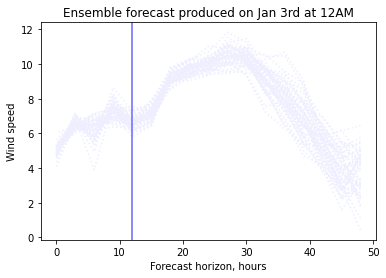}
   %\hfill
    %\includegraphics[width =0.48\textwidth]{pic/real.png}
\end{figure}

\subsection{Model calibration}\label{methodcali}

To calibrate our models for forecast dynamics from meteorological ensemble forecasts, we use an approach inspired by the EMOS methodology in \cite{gneiting2005calibrated}, to determine the conditional mean and variance of the predictive distribution from the ensemble forecasts. As explained in the introduction, the ensemble forecasts may be biased and underdispersed, so that the mean and variance of the predictive distribution are not necessarily equal to the mean and variance of the empirical distribution of the forecast members, although these quantities are certainly related to each other. 
  Let $x^m_{htk}$ denote the value of member $m$ of the ensemble forecast of a given meteorological quantity (wind speed or temperature), recorded at time $t$, at location $k$, for the forecast horizon $h$, and by $\tilde x_{tk}$ the corresponding realization. We assume that the mean of the predictive distribution, denoted by $m_{htk}$ is a linear function of the mean of ensemble members:
  \begin{align}m_{htk} =  m_{htk}(a^0_h,a^1_h) = a^0_h + \frac{a^1_h}{M}\sum_{m = 1}^M x^{m}_{htk} := a^0_h + a^1_h \bar x_{htk}.\label{emosmean}
 \end{align}   
  The coefficients $a^0_h$ and $a^1_h$ reflect the bias in the ensemble forecasts. In the case of unbiased forecasts we would have $a^0_h= 0$ and $a^1_h = 1$. 
  Similarly, the variance of the predictive distribution  depends on the spread of ensemble members, but the latter may not reflect the forecasting error entirely. Hence we assume that at each date $t$, each location $k$, and each lead time $h$, the variance of the predictive distribution is given by
  \begin{align}\sigma^2_{htk} = \sigma^2_{htk} (c_h,d_h) = c_h+d_h \frac{1}{M}\sum_{m = 1}^M (x^m_{htk}-\bar x_{htk})^2 := c_h+d_hV_{htk}.\label{emosvar}
  \end{align}
 The coefficients in the expression for the variance can be interpreted as follows: $c_h$ represents the part of the error that is not related to the spread of the ensemble members, whereas $d_h$ represents the part of the conditional variance explained by the ensemble spread.

 The full model specification for the temperature forecasts is thus given by equations (\ref{nig1}--\ref{nig.dist}) and (\ref{emosmean}--\ref{emosvar}), while the full specification for the wind forecasts is given by equations (\ref{pos1}--\ref{logNig.dist}) and  (\ref{emosmean}--\ref{emosvar}).

 The full model is calibrated in a three-step procedure as detailed below.
 \paragraph{Step 1} In the first step, we first calibrate, separately for each forecast horizon, the parameters $a^0_h$ and $a^1_h$ by linear regression:
 $$  (\hat a_h^0,\hat a_h^1) = \underset{a^0,a^1}{\arg\min} \sum_{t,k = 1}^{T,K} (a^0 + a^1 \bar x_{htk}-\tilde x_{tk})^2.$$
 Next, using the calibrated values $\hat a_h^0$ and $\hat a_h^1$, we calibrate $c_h$, $d_h$ and $b_h$ by maximum likelihood:
 $$
 (\hat c_h,\hat d_h,\hat b_h) = \underset{c,d,b}{\arg\max} \sum_{t,k = 1}^{T,K} \log p(\tilde x_{tk}, m_{htk}(\hat a_h^0,\hat a_h^1),\sigma^2_{htk}(c,d),b),
 $$
 where $p(x, m,\sigma^2,b)$ denotes the predictive density expressed in terms of the conditional mean $m$, the conditional variance $\sigma^2$ and the shape parameter $b$. 

For the {temperature}, the predictive density \eqref{nig.dist} writes,
$$
p^\tau(x,m,\sigma^2,b) =\frac{\frac{\sigma^2}{b^2} K_1\left(\frac{1}{b} \sqrt{(\sigma^2/b)^2 +
      (x-m)^2}\right)}{\pi \sqrt{(\sigma^2/b)^2+(x-m)^2}}
  e^{\sigma^2/b^2 },
  $$
  and for the wind speed, the predictive density \eqref{logNig.dist} writes,
  \begin{equation*}
   p^w(x,m,\sigma^2,b) = \frac{\alpha V K_1\left(\frac{\alpha}{b}\sqrt{V^2 +
      b^2(x-\log m)^2}\right)}{\pi \sqrt{V^2+b^2(x-\log m)^2}}
  e^{V(b^{-2} + \frac{1}{2}) - \frac{1}{2}(x-\log m)} 
\end{equation*}
where 
$$
\alpha  =  \sqrt{ (b^{-1}+\frac{b}{2})^2  + \frac{1}{4}}\quad \text{and}\quad 
V = \log\left(\frac{\sigma^2}{m^2}+1\right). 
$$

Note that in this step, the shape parameter $b_h$ is calibrated independently for each forecast horizon. A common value for all horizons will be fixed in the next step.
The choice of the maximum likelihood procedure in this first step was motivated by the availability of predictive densities in explicit form. An alternative would be to use the Continuous Ranked Probability Score (CRPS) as in \cite{gneiting2005calibrated}, but for our models the CRPS is only available through heavy numerical computation, making the approach of \cite{gneiting2005calibrated} difficult to implement. We also tested direct maximum likelihood estimation of the four parameters $a^0_h,a^1_h,c_h,d_h$ and $b_h$, but the presented approach where linear regression is used for $a^0_h,a^1_h$ leads to better results.

\paragraph{Step 2}
In the previous step, the shape parameter $b$ was calibrated separately for each lead time. However, in our model, this parameter does not depend on the forecast lead time. Thus, once the parameters $a^0_t, a^1_t, c_t, d_t$ have been estimated in the first step, we perform again an estimation of the parameter $b$ by maximizing the likelihood including all forecast horizons:
$$\hat b = \underset{b}{\text{argmax}} \sum_{h \in \mathcal H}\sum_{t,k = 1}^{T,K} \log p(\tilde x_{tk}, m_{htk}(\hat a_h^0,\hat a_h^1),\sigma^2_{htk}(\hat c_h,\hat d_h),b).$$
This formulation applies for both the temperature and the wind speed calibration. This procedure does not impact the goodness of fit in terms of first and second moment since the parameters of the mean and the variance are fixed. However, the shape of the distribution may change a bit with no major impact.%\todo{: mettre un graphique. +CRPS}

\paragraph{Step 3}
Once we have estimated the 'static' properties of the model, that is, the parameters which appear in the predictive distribution ($b$ and $a_0$, $a_1$, $c$ and $d$ for each time horizon), we need to estimate the 'dynamic parameter', that is, the function $\rho$, which describes how the forecast varies dynamically. While the static parameters are estimated by comparing the forecasts with their respective realizations, $\rho$ can be estimated by comparing forecasts for the same quantity, obtained at different dates. 

For the temperature model, from equation (\ref{nig2}), we may write: 
$$\mathbb E\left[\frac{V_{t+s}}{V_{t}}|\mathcal F_t\right] = \exp{\left(-\int_t^{t+s} \rho(T-u)^2du\right)}. $$
Since the expectation in the right-hand side is deterministic, we can remove the conditioning and write:
\begin{align}\int_t^{t+s} \rho(T-u)^2du = -\log\left( \mathbb E\left[\frac{V_{t+s}}{V_{t}}\right]\right)\label{estimrho}
 \end{align}
Based on this identity, we suggest the following approach to calibrate $\rho$ based on a given discrete set of forecast horizons $0<h_1<\dots<h_H=T$, for which data are available. Here $h_H$ corresponds to the longest available horizon (when the simulation starts) and $h_1$ corresponds to the shortest available horizon (last available forecast for a given realization). 
Assume that $\rho$ is constant on the intervals $(h_i,h_{i+1})$, $i=1,\dots,H-1$ and denote the value of $\rho$ on the interval $(h_i,h_{i+1})$ by $\rho_i$. This assumption is without loss of generality since, in view of Equations (\ref{nigtc1}--\ref{nigtc2}) and (\ref{pos1tc}--\ref{pos2tc}), the law of $(m_{h_{i+1}},V_{h_{i+1}})$ conditional on $(m_{h_{i}},V_{h_{i}})$ depends on $\rho$ only through the integral $\int_{h_i}^{h_{i+1}} \rho(T-u)^2du$. 

In view of \eqref{estimrho}, we propose to estimate the values $\rho_1,\dots,\rho_{H-1}$ as follows:
  $$
 \hat \rho^2_i = -\frac{1}{ (h_{i+1}-h_i)}\log \frac{1}{TK  }\sum_{t,k=1}^{T,K}\frac{\sigma^2_{h_{i+1}tk}(\hat c_{h_{i+1}},\hat d_{h_{i+1}})}{\sigma^2_{h_{i}tk}(\hat c_{h_{i}},\hat d_{h_{i}})}.
 $$

 For the wind speed model, from equations (\ref{pos1}-\ref{pos2}), we may write: 
%$$\mathbb E\left[\frac{V_{t+s}}{V_{t}}|\mathcal F_t\right] = \exp{\left(-\left(1+\frac{b^2}{2}\right)\int_t^{t+s} \rho(T-u)^2du\right)}, $$
%so that
%\begin{align}\int_t^{t+s} \rho(T-u)^2du = -\frac{2}{2+b^2}\log\left( \mathbb E\left[\frac{V_{t+s}}{V_{t}}\right]\right),\label{estimrhow}
%\end{align}
%where now
%$$
%V_t = \log\left\{1+\frac{\text{Var}[m_T|\mathcal F_t]}{m_t^2}\right\}.
%$$
%As above, we assume that $\rho$ is piecewise constant and suggest the following estimator:
%$$
%\hat \rho^2_i = -\frac{2}{(2+\hat b^2) (h_{i+1}-h_i)}\log \frac{1}{TK  }\sum_{t,k=1}^{T,K}\frac{\log\left\{1+\frac{\sigma^2_{h_{i+1}tk}(\hat c_{h_{i+1}},\hat d_{h_{i+1}})}{m^2_{h_{i+1}tk}(\hat a^0_{h_{i+1}},\hat a^1_{h_{t_{i+1}}})}\right\}}{\log\left\{1+\frac{\sigma^2_{h_{i}tk}(\hat c_{h_{i}},\hat d_{h_{i}})}{m^2_{h_{i}tk}(\hat a^0_{h_{i}},\hat a^1_{h_{t_{i}}})}\right\}}.
%$$
$$
\mathbb E\left[\log\frac{m_{t+h}}{m_t}|\mathcal F_t\right] = -\frac{1}{2}\int_t^{t+h}\rho^2(T-s) \mathbb E[V_s|\mathcal F_t] ds,
$$
and 
$$
\mathbb E[V_s|\mathcal F_t] = V_t \exp\left(-\left(1+\frac{b^2}{2}\right) \int_t^s \rho^2(T-u)du\right),
$$
so that
$$
\mathbb E\left[\log\frac{m_{t+h}}{m_t}|\mathcal F_t\right] = -\frac{V_t}{2+b^2} \left(1-\exp\left(-\left(1+\frac{b^2}{2}\right) \int_t^{t+h} \rho^2(T-u)du\right)\right).
$$
Dividing both sides by $V_t$, we can remove the conditioning and rewrite this expression as follows:
$$\int_t^{t+h} \rho^2(T-u)du = - \frac{2}{2+b^2} \log\left(1+(2+ b^2)\mathbb E\left[\frac{1}{V_t}\log\frac{m_{t+h}}{m_t}\right]\right).$$
As before, assume without loss of generality that $\rho$ is constant on every interval $(h_i,h_{i+1})$ for $i=1,\dots,H-1$, and denote its value on such interval by $\rho_i$. This suggests the following estimator for $\rho_i$:
$$
\hat \rho_i = -\frac{1}{h_{i+1}-h_i}\frac{2}{2+\hat b^2}\log\left(1+(2+ \hat b^2)\frac{1}{TK}\sum_{t,k=1}^{T,K}\frac{1}{\sigma^2_{h_i tk}(\hat c_{h_i},\hat d_{h_i})}\log\frac{m_{t+h}}{m_t}\right).
$$
%{\color{blue}Here we use the trick: %$\log(x)-x+1<0$. This is not possible for %the temperature since we get $$ \mathbb %E\left[\log m_{t+h}-\log m_t %|\mathcal{F}_t\right] =-\frac{1}{2} %\int_t^{t+h}\mathbb %E\left[\frac{V_s}{m^2_s} %\bigg|\mathcal{F}_t\right] \rho(T-s)^2ds $$
%How to compute $\mathbb %E\left[\frac{V_s}{m^2_s} %\big|\mathcal{F}_t\right]$? }

\subsection{Numerical illustrations}
\label{stats}
In this section we apply the methodology presented in section \ref{methodcali} to ensemble forecasts for the wind speed and the temperature  described in \ref{datapres}. %We recall that we use the data available until the 21$^{st}$ of January to calibrate our model and the data from the 22$^{nd}$ to the 31$^{st}$ of January to test the predictive performances of it.
%We first analyse the calibration results, test its accuracy and finally confront the realisations to the predictions. 

\paragraph{Estimated coefficients}
%Table \ref{units} sums up the units for the coefficients for the wind speed and the temperature. 
%\begin{table}[ht]
%   \centering
%\begin{tabular}{|c|c|c|}
%\hline
%Coefficients &      Wind speed &       %Temperature \\
%\hline
% $a^0,a^1$ &  ln($m.s^{-1}$) &  $\celsius$\\
%$c,d$  & ln($m.s^{-1}$)$^{2}$&   $\celsius^2$ \\
%$b,\rho$ & ln($m.s^{-1}$)$.s^{-1}$ &  %$\celsius.s^{-1}$\\
%\hline
%\end{tabular} 
%    \caption{Coefficients units for the wind %speed and the temperature}
%    \label{units}
%\end{table}
As explained in section \ref{methodcali}, we use an EMOS-inspired technique to improve the calibration of ensemble forecasts. To motivate this post-processing, we present in Figure \ref{Fig6} the Talagrand diagrams (rank histograms) for the ensemble forecasts of the log wind speed and the temperature, constructed using the test data. Talagrand diagram is a tool for checking the quality of calibration of ensemble forecasts and is a histogram of the ranks of observations within the corresponding forecast ensembles. 
 In other words, for a given forecast horizon $h$, we plot the histogram of $R(\tilde x_{tk},(x_{htk}^m)_{m = 1\dots M })_{t,k=1}^{T,K}$, where $R(\tilde x, (x^m)_{m = 1\dots M })$ is the normalized rank of the observation $\tilde x$ within the ensemble $(x^m)_{m = 1\dots M }$. For a perfectly calibrated ensemble forecast, the Talagrand diagram is within the confidence bounds of the uniform distribution. In the present case, histograms in Figure \ref{Fig6}, for the lead time $12h00$, and Figure \ref{Fig62}, for the lead time $24h00$, present a U-shaped profile, which is a clear indication of under-dispersion of our forecast ensembles. In addition, the asymmetric form of the diagram for the wind speed in Figure \ref{Fig6}, and for the temperature in Figure \ref{Fig62}, is an indication of the presence of a bias in the ensemble forecast. Histograms for lead times $36h00$ and $48h00$ are available in Appendix \ref{HistoAppend}. %\todo{Concernant l'intervalle de confiance. Je ferais simplement un intervalle de confiance autour de la droite pointillée, correspondant à la situation ou la loi est uniforme. L'intervalle de confiance autour de l'estimateur est plus délicat, d'autant plus qu'on ne peut pas dire que nos données sont i.i.d. Si on trace tout simplement, à titre indicatif, un intervalle de confiance uniforme, on sera plus difficilement attaquable par un reviewer. Aussi, en anglais on dit 'theoretical' pas 'theoretic'}
 
%\todo[inline]{Dans Figure 8, 9, 10, 11, 16-19 il faut remplacer Theoretic par Theoretical et Confidance par Confidence}

\begin{figure}[ht]
      \centering
      \caption{Talagrand diagrams for the wind speed and the temperature, lead time $12h00$}
      \includegraphics[width = 0.65\textwidth]{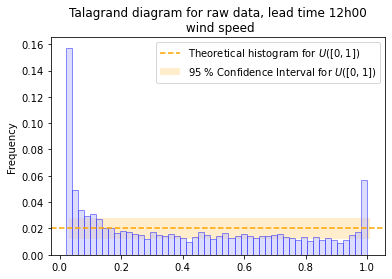}\label{Fig6}
      \includegraphics[width = 0.65\textwidth]{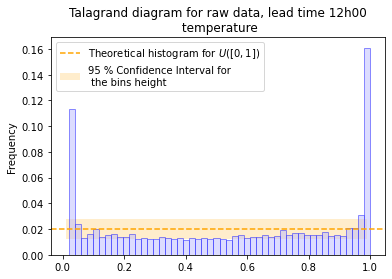}
   %\hfill
    %\includegraphics[width =0.48\textwidth]{pic/real.png}$
    \label{talagrand1}
\end{figure}
\begin{figure}[ht]
      \centering
      \caption{Talagrand diagrams for the wind speed and the temperature, lead time $24h00$}
      \includegraphics[width = 0.65\textwidth]{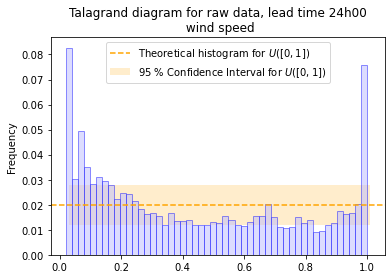}\label{Fig62}
      \includegraphics[width = 0.65\textwidth]{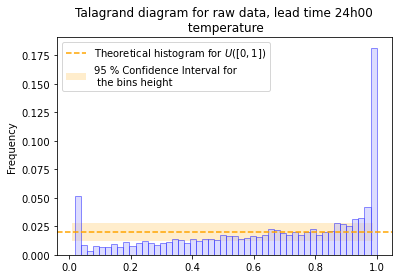}
          \label{talagrand2}

   %\hfill
    %\includegraphics[width =0.48\textwidth]{pic/real.png}
\end{figure}
In view of the Talagrand diagrams discussed above, we apply our post-processig approach to obtain an unbiased and well-calibrated probabilistic forecast. Tables \ref{windtab} and  \ref{temptab} show, respectively, the estimated  coefficients for the wind speed and the temperature during the training period.

\begin{table}[ht]
   % \centering
\begin{tabular}{|c|c|c|c|c|c|}
\hline
Lead time &      $a^0$ &       $a^1$ &         $c$ &         $d$ &     $b$ \\
\hline
12h  &  0.117 &  0.964 &  0.360 &  0.765 &  0.035\\
24h & -0.028 &  0.994 &  0.446 &  0.494 &  0.035 \\
36h  & -0.083&  1.006 &  0.951 &  0.160 &  0.035 \\
48h  & -0.573 &  1.044&  0.240 &  0.798 &  0.035 \\
\hline
\end{tabular} 
\hspace{0.5cm}
\begin{tabular}{|c|c|}
\hline
Lead time & $\rho(.)$\\ interval  &  \\ 
\hline
12h-24h &  0.171\\
24h-36h &  0.153 \\
36h-48h &  0.168\\
\hline
\end{tabular}
    \caption{Calibrated coefficients for the wind speed from the 1$^{st}$ to the 21$^{th}$ of January 2015 }
    \label{windtab}
\end{table}

\begin{table}[ht]
    \begin{tabular}{|c|c|c|c|c|c|}
\hline
Lead time  &       $a^0$ &       $a^1$ &         c &         d &     b \\
\hline
12h  &  0.217 &  0.952 &  0.312 &  1.722 &  0.719 \\
24h &  0.103&  1.006 &  0.416 &  0.916 &  0.719 \\
36h  &  0.136 &  1.019 &  0.675 &  0.467 &  0.719 \\
48h &  0.202 &  1.015 &  0.306 &  0.913 &  0.719 \\
\hline
\end{tabular}
\hspace{0.5cm}
\begin{tabular}{|c|c|}
\hline
Lead time & $\rho(.)$\\ interval  &  \\ 
\hline
12h-24h &  0.160 \\
24h-36h &  0.163\\
36h-48h &  0.180 \\
\hline
\end{tabular}
    \caption{Calibrated coefficients for the temperature from the 1$^{st}$ to the 21$^{th}$ of January 2015 }
    \label{temptab}
\end{table}

%The bias coefficients $a_0, a_1$ confirm the bais observed in the diagrams. There is a clear link between the realisation and the ensemble forecasts but the former is not totally explained by the latter. For the wind speed at the $48h$ lead time we note that $a_0 = -0.573$ and  $a_1= 1.044$ which indicates that the observed overestimation of the wind speed at the lead time $12h$ is also pronounced for other lead times.\todo{Tout cela n'est pas très clair pour moi. Pourquoi $a^0>0$ pour 12h alors que les valeurs sont négatives pour les autres lead times? Comment que tu peux conclure si le biais est positif ou négatif avec ces valeurs des coefficients? } 

The spread coefficients $c$ and $d$  confirm that the raw ensemble forecasts are underdispersed: except for the $24h$ lead time wind speed forecasts, all provide a post process variance larger than the ensemble forecast spread. This is especially true for the $12h$ lead time temperature forecast which has an intercept of $c = 0.312$ and a slope of $d = 1.772$ which multiply almost by two the original variance. %\todo{Peut-être il faut donner les diagrammes de Talagrand aussi pour un autre lead time, 24 ou 48 par exemple, car le comportement à l'air d'être très différent pour 12h}

We now proceed to analyse the diffusion coefficients $b$ and $\rho$. The parameter $b$ is quite low for the log wind speed distribution. This suggests that the log-wind speed distribution is closer to the Gaussian one, than the temperature distribution for which the coefficient $b$ is higher. This indicates a heavier tailed model where extreme temperature values are more likely to happen than extreme wind spikes.

The values of the piecewise constant function $\rho(.)$ are very close to each other on the three time intervals considered for both the wind speed and the temperature. This feature will be useful when using this dynamics in the control problems in section \ref{controlpb}.%\todo{Je ne sais pas trop comment interpreter plus cette fonction} 

\paragraph{Goodness of fit}\label{GoF}

In this paragraph, we check the goodness of fit of the estimated predictive distribution using the test dataset and provide some illustrations. 
We first present the mean square error computed using the test data before and after the pre-processing that is :
$$\text{MSE}^{\text{raw}}(h) = \sum_{t,k=1}^{T,K} (\bar x_{htk}-\tilde x_{tk})^2,\quad \text{MSE}(h) = \sum_{t,k=1}^{T,K} (a^0_h+a^1_h \bar x_{htk}-\tilde x_{tk})^2.$$

Except for the wind speed at lead times $36h$ and $48h$, the correction of the bias brings the mean of the test data closer to the realization. Hence our estimation procedure doesn't overfit the training period and is robust when considering new data. However, we should mention that this is done over data for the same month and the same season. We may assume that the seasonality impacts the value of the coefficients and repeating the study over different months may provide additional insights.

\begin{table}[ht]
    \centering
\begin{tabular}{|c|c|c|c|}
\hline 
\multicolumn{3}{|c|}{Log wind speed} \\
\hline
 Lead time&  MSE &    MSE\\
 &  Raw ensemble & Model \\
\hline
12h        &      0.948 &  0.804  \\
24h        &      0.955 &  0.931  \\
36h        &      1.233 &  1.240 \\
48h        &      1.633 &  1.795  \\
\hline
\end{tabular}
\hspace{0.5cm}
\begin{tabular}{|c|c|c|}
\hline
\multicolumn{3}{|c|}{Temperature}\\
\hline
 Lead time&  MSE &    MSE \\
 &  Raw ensemble & Model \\
\hline
12h        &      0.656 &  0.599 \\
24h        &      0.851 &  0.778  \\
36h        &      0.973 &  0.900  \\
48h        &      1.539 &  1.424  \\
\hline
\end{tabular}
    \caption{MSE for the log wind speed and the temperature over the period (22/01-31/01)}
    \label{tab:temp}
\end{table}

To evaluate the calibration of the predictive distribution we use the  probability integral transform (PIT) histogram and to check both calibration and sharpness, we compute the continuous rank probability score (CRPS). 
The formal definition of the CRPS for  a given realisation $y$ and a predictive distribution with cumulative distribution function (CDF) $F$ is given by:%\todo{completer la definition}
      $$\text{CRPS}(F,y)= \int_{\mathbb R} (F(x)-\mathbbm{1}_{\{y\leq x\}})^2dx $$
For the normal inverve Gaussian and the log normal inverse Gaussian distributions, it is not possible to compute the analytical expression of the CDF. However, we can use the Plancherel formula and obtain an expression relying on the characteristic function $\phi$ of the normal inverse Gaussian  predictive distribution: 
\begin{align*}
\int_{\mathbb R}(F(x)-\mathbf 1_{y\leq x})^2 dx &= \frac{1}{2\pi} \int_{\mathbb R} \frac{|\phi(u)-e^{iuy}|^2}{u^2} du \\ &= \frac{1}{2\pi} \int_{\mathbb R} \frac{\Big|e^{i \mu u + \delta \left(\gamma-\sqrt{\alpha^2-(\beta + i u)^2}\right)}-e^{iuy}\Big|^2}{u^2} du,
\end{align*}

Following this formulation, for a predictive distribution for the location $k$, the lead time $h$ and the date $t$, we denote the CDF $F_{tkh}$ and the realisation $y_{tkh}$. The parameters
 $\alpha, \beta, \gamma, \delta$ and $\mu$ for temperature are given by, 
 $$ \alpha = \frac{1}{b}, \; \beta = 0, \; \gamma = \frac{1}{b},\; \delta = \frac{V_{htk}}{b}, \; \mu = m_{htk},$$
 and for the wind speed:
 $$ \alpha = \sqrt{\left(\frac{1}{b}+\frac{b}{2}\right)^2 +\frac{1}{4}}, \; \beta = -\frac{1}{2}, \; \gamma =\frac{1}{b}+\frac{b}{2} ,\; \delta = \frac{V_{htk}}{b}, \; \mu = \log(m_{htk}).$$

%a numerical approximation given by:
%\begin{equation}
 %   \hat F_h(x) = \sum_{i = 1}^n f_h(x_i) \Delta 
%\end{equation}
%where $f_h$ is the predictive density calibrated for the lead time $h$, $ x_0 < \dots< x_i< \dots< x, \forall x \in \mathbb R, \; \Delta = x_{i+1}-x_i, i = 1, \dots,n$ and $n\Delta  = x $. 
We analyse the goodness of fit using the averaged CRPS over the period 22/01/15-31/01/15 for each lead time,  
$$\text{averageCRPS}_h(F, y) = \frac{1}{TK}\sum_{t,k= 1}^{T,K}\int_\mathbb R \left( F_{tkh} (x) -\mathbbm{1}_{\{y_{tkh}\leq x\}}\right)^2dx$$
%{\color{magenta}I don't understand how you compute the CRPS, in particular, the outer integral. One may simplify this formula using the Fourier transform (Plancherel formula) if we know the characteristic function $\phi$ of the predictive distribution, which is the case for NIG:
%$$
%\int_{\mathbb R}(F(x)-\mathbf 1_{y\leq x})^2 dx = \frac{1}{2\pi} \int_{\mathbb R} \frac{|\phi(u)-e^{iuy}|^2}{u^2} du. 
%$$}

We compare it to the averaged CRPS obtained with the raw ensemble forecasts, that is:

 $$ \text{averageCRPS}_h(\text{Ensemble forecasts}, y_{tkh}) = \frac{1}{TK}\sum_{t,k = 1}^{T,K} \int_\mathbb R\left(\hat F^M_{tkh}(x)-\mathbbm{1}_{\{y_{tkh}\leq x\}}\right)^2dx, $$
where $$
\hat F^M_{tkh}(x) = \frac{1}{M}\sum_{m=1}^M \mathbf 1_{\{x^m_{tkh} \leq x\}}
$$
and for each location, date and lead time, we may simplify the CRPS formula as follows: 
\begin{align*}
\text{CRPS}(\hat F^M_{tkh},y_{tkh}) &=\int_\mathbb R\left(\hat F^M_{tkh}(x)-\mathbbm{1}_{\{y_{tkh}\leq x\}}\right)^2dx \\ &= \frac{2}{M}\sum_{\ell=1}^M (x_{(\ell)}-y_{tkh})\left\{\mathbf 1_{\{x_{(\ell)}>y_{tkh}\}} - \frac{\ell-\frac{1}{2}}{M}\right\},
\end{align*}
where $x_{(\ell)},k=1\dots M$ is the order statistics of the sample $x^m_{tkh},m=1\dots M$.
%{\color{magenta}You should use the second method. There is a simplified formula:
%When the predictive distribution is the empirical distribution
%  of values $(x_1,\dots,x_n)$, that is,
%$$
%\hat F_n(x) = \frac{1}{n}\sum_{i=1}^n \mathbf 1_{x_i \leq x}
%$$
%then 
%\begin{align*}
%crps(\hat F_n,y) = \frac{2}{n}\sum_{k=1}^n (x_{(k)}-y)\left\{\mathbf 1_{x_{(k)}>y} - \frac{k-\frac{1}{2}}{n}\right\},
%\end{align*}
%where $x_{(k)},k=1\dots n$ is the order statistics.
%}

\begin{table}[ht]
\caption{CPRS for the wind speed ($ \log m.s^{-1}$) and for the temperature (C$\degree$) over the test period 22/01/15-31/01/15}
    \centering
   \begin{tabular}{|c|c|c|c|c|} 
   \hline
\multicolumn{5}{|c|}{log wind speed}\\
  \hline
   {Lead time } &        12h &        24h &       36h &        48h \\ 
   \hline 
   EMOS &   0.039& 0.056 &   0.062 &  0.128\\
   \hline 
   Raw  &  0.080 & 0.095 &  0.100 &  0.102\\
   \hline
\end{tabular}
    \centering
    \begin{tabular}{|c|c|c|c|c|}
       \hline
\multicolumn{5}{|c|}{temperature}\\
    \hline
    Lead time &        12h &        24h &        36h &        48h \\
        \hline
EMOS & 0.372&  0.437 &  0.485 &  0.639\\
    \hline
Raw  &  0.450 &  0.481 &  0.538 &  0.752 \\
    \hline
\end{tabular}
    \label{tab:my_label}
\end{table}

Table \ref{tab:my_label} compares the CRPS computed using the test period for raw ensemble forecast and for probabilistic forecasts obtained using our post-processing method. We observe a significant improvement for both wind speed and temperature forecasts, for all forecast horizons except the 48-hour forecast horizon for the wind speed.  

As an independent illustration of the calibration of the post-processed forecasts we plot the probability integral transforms (PIT), which is the equivalent of Talagrand diagram in the context of probabilistic forecasts and consists in plotting the histogram of the predictive CDF evaluated at the realization point. If the predictive distribution is well calibrated then the histogram should be close to the uniform one.

\begin{figure}[ht]
      \centering
      \caption{PIT histogram for the wind speed and the temperature, lead time $12h00$ }
      \includegraphics[width = 0.48\textwidth]{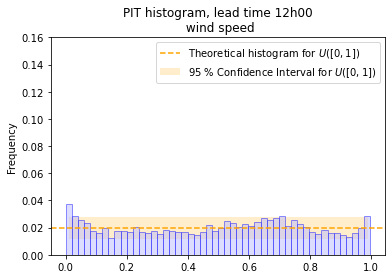}\label{Fig7}
      \includegraphics[width = 0.48\textwidth]{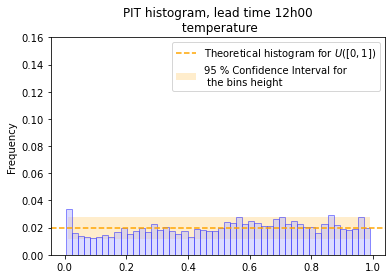}
   %\hfill
    %\includegraphics[width =0.48\textwidth]{pic/real.png}
\end{figure}

The PIT of post-processed forecasts are shown in  Figures \ref{Fig7} and \ref{Fig72} and may be compared to Talagrand diagrams in Figures \ref{talagrand1}--\ref{talagrand2}. Here again, the improvement is considerable, although some deviations from the uniform distribution can still be observed. They may be explained by the fact that we use a parametric approach, which obviously cannot provide a perfect fit to the data, and our observations are not completely independent.  

\begin{figure}[ht]
      \centering
      \caption{PIT histogram for the wind speed and the temperature, lead time $24h00$}
      \includegraphics[width = 0.48\textwidth]{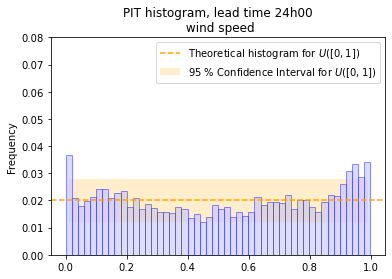}
      \includegraphics[width = 0.48\textwidth]{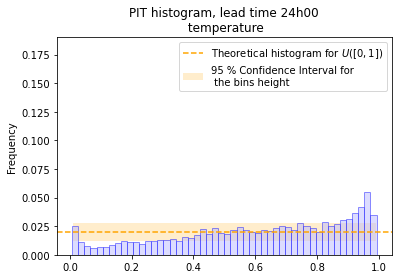}\label{Fig72}
   %\hfill
    %\includegraphics[width =0.48\textwidth]{pic/real.png}
\end{figure}

%The quality of the estimated predictive distributions is a key factor of the accuracy of the prediction and the uncertainty one has on it. 
\newpage
\paragraph{Behavior of the predictive density in test data}
To illustrate the shape of the predictive density obtained with our approach, we displayed in Figure \ref{Fig3}, for a given realization date (22/01/15 at 12 a.m) and location, the predictive densities at each lead time as well as the realisation for the wind speed forecasts. 
We observe that the sharpness of the predictive distribution varies with the lead time. Interestingly it doesn't always improve as we approach the realisation time (e.g see the temperature predictive densities at lead times $12h$ and $24h$). This is to be compared with the simulated predictive distributions in Figures \ref{lnigci} and \ref{nigci}. 
%\todo{changer le titre}
%\todo[inline]{Dans Figure 12, on ne peut pas distinguer les différentes densités, les couleurs sont trop proches.}
\begin{figure}[ht]
      \centering
      \caption{Example of predictive densities produced by the model}
      \includegraphics[width = 0.47\textwidth]{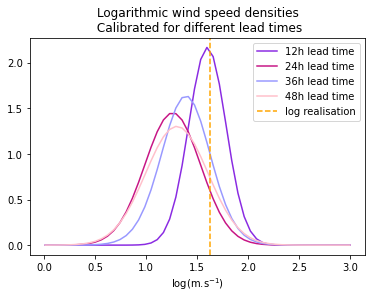}\label{Fig3}
      \includegraphics[width = 0.47\textwidth]{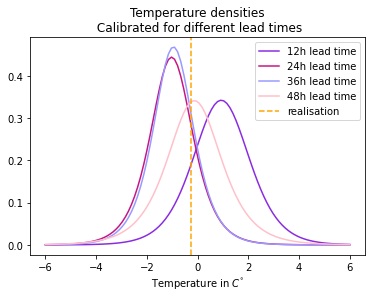}
   %\hfill
    %\includegraphics[width =0.48\textwidth]{pic/real.png}
\end{figure}

At the same location we plotted the evolution of the point forecast (first moment of the predictive distribution) for a fixed lead time and the corresponding realization over the period 22/01/15-31/01/15 in Figure \ref{Fig2}. We also show the confidence intervals around the point forecast. The realization always falls in the 90 $\%$ confidence interval and the width of the intervals varies throughout the simulation. The forecasts seem reasonably close to the realizations: this is especially true when the width of the confidence interval is small. 

%\todo{Dans Figure 13, dans Realization vs. model, le point est après s et non avant.}
\begin{figure}[ht]
      \centering
      \caption{Temperature and wind speed forecasts for the period 22/01/2015-31/01/2015}
      \label{Fig2}
    \includegraphics[width = 0.45\textwidth]{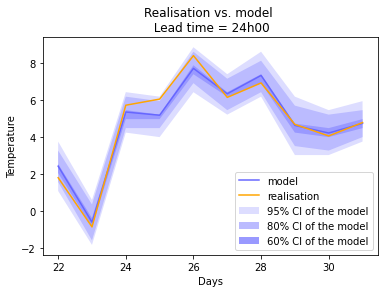}
      \includegraphics[width =0.45\textwidth]{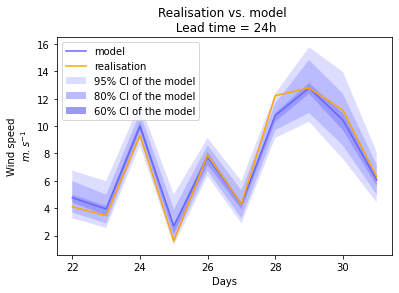}
      \includegraphics[width = 0.45\textwidth]{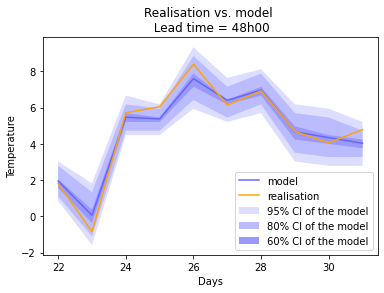}
      \includegraphics[width =0.45\textwidth]{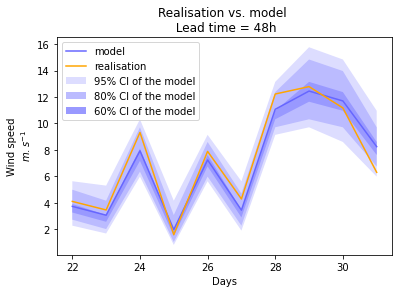}
   %\hfill
    %\includegraphics[width =0.48\textwidth]{pic/real.png}
\end{figure}

 Table \ref{CItraining} shows  the average width of the 90 $\%$ confidence interval, for the testing period: 22/01/2015-31/01/2015. 
For the first three lead times the width increases with lead time but the 36h lead time confidence intervals are on average sightly larger than the 48 h lead time ones. This confirms the intuition that the closer we get to the realisation date the less variations there are in forecasts updates (on average). On the other hand, finding a wider confidence interval average for the lead time $36h$ than for the lead time $48h$ is quite unexpected. We investigated the confidence interval width during the training period in Table \ref{tab:CI} and found that the confidence interval width always decreases as we approach the realization time. This suggests that the model may not be fully consistent with the test data, perhaps owing to a possible non-stationarity of the forecasts.

\begin{table}[ht]
    \centering
\begin{tabular}{|c|c|c|c|}
\hline
 Lead time&  90$\%$ CI width  & 90$\%$ CI width\\
 &Wind speed & Temperature \\
\hline
12h &       2.565 &2.605 \\
24h        &     2.869& 2.923\\
36h       &    3.381 &  3.301\\
48h        &      3.628 & 3.585\\
\hline
\end{tabular}
    \caption{Average 90 $\%$ confidence interval for the period 01/01-21/01 for the wind speed and the temperature}
    \label{tab:CI}
\end{table}

\begin{table}[ht]
    \centering
\begin{tabular}{|c|c|c|c|}
\hline
 Lead time&  90$\%$ CI width  & 90$\%$ CI width\\
 &Wind speed & Temperature \\
\hline
12h &       2.537 &2.592 \\
24h        &     2.801 & 2.747\\
36h       &    3.378 &  3.062\\
48h        &      3.254 & 2.956\\
\hline
\end{tabular}
    \caption{Average 90 $\%$ confidence interval for the period 22/01-31/01 for the wind speed and the temperature}
    \label{CItraining}
\end{table}

\newpage
\section{Application to wind power trading}\label{controlpb}
In this section we present an application of our methodology based on the dynamic modeling of probabilistic forecasts to the problem of wind power trading in the intraday electricity market.

\subsection{Description of the problem}\label{windtrading}

Consider a wind power producer who aims to sell the output power in
the intraday electricity market. To analyze the effect of market mechanisms we assume that there are no subsidies and no guaranteed purchase scheme. %The producer can
%either sell the generated power at the time of production (that is, in
%the intraday market where transactions can be made 15 or 30 minutes
%prior to delivery), or in advance, using the intraday market with
%longer-dated contracts.  
The intraday market opens every day at 3 p.m and allows continuous trading in all  delivery hours of the next day. For a given delivery hour $T$, we consider the energy produced during a small time interval
around this date $T$. We denote the average wind speed during this time
interval by $m_T$, and the power curve of the wind turbine by $f$,
that is, the rate of power production during this interval is given by
$P_T = f(m_T)$. For the purpose of illustration we choose the stylized production function $f$ defined by: 
$$f(m) = \frac{(m-m_{\min})^+ - (m - m_{\max})^+ }{m_{\max}-m_{\min}},$$
where  $m_{\min}$ is the cut-in speed (at which the turbine starts to produce), and $m_{\max}$ is the rated speed (at which the turbine produces its maximum power), but the methodology applies without modifications to any other production function. \\

This power can be sold at any time starting from the opening time of the intraday market up to 15 minutes before production.  The fraction traded at the date
$t$ will have the price $S_t$, and we denote the total amount of power (for delivery at $T$) sold or bought up to date $t$ by $\phi_t$. Any power not sold in the intraday market prior to date $T$ will be sold at date $T$ at the balancing price denoted by $S_T$. {In addition, balancing transactions are subject to imbalance penalty equal to a constant $K$ times the volume of the transaction.}
Throughout this section, we assume the following dynamics for the price:
\begin{align}\label{eq:PriceProcess}
    dS_t = \mu_S dt + \sigma_S dB_t, \quad \forall t \in [0,T]
\end{align}
where $\mu_S$ and $\sigma_S$ are constants and $(B_t)_{t\in [0,T]}$ is a Brownian motion.

 We make the assumption that the producer changes her position in the market only when a new forecast of the wind speed, and thus of the power production, becomes available. In other words, new trades are only triggered by new forecast information and not by price information which is available continuously. This is justified by the fact that most producers do not attempt to take advantage of potential price arbitrages but use the markets to compensate forecast errors. We denote by $t_0<\dots<t_{N-1}$ the discrete times at which the trades take place, $t_N = T$ being the time when the delivery starts. The profit of the producer is thus given by
 \begin{align*}
 \text{Profit} &= \underbrace{S_{t_0} \phi_{t_0} + \sum_{i=1}^{N-1} S_{t_i} (\phi_{t_i} - \phi_{t_{i-1}})  }_{\text{Intraday market}}+ \underbrace{S_T(f(m_T) - \phi_{t_{N-1}}) - K|f(m_T) - \phi_{t_{N-1}}|}_{\text{Imbalance payment}}\\
 & = f(m_T)S_T - \sum_{i= 0}^{N-1}
\phi_{t_i} \Delta S_{t_i}-K|f(m_T) - \phi_{t_{N-1}}|,
 \end{align*}
 where $\Delta S_{t_i} = S_{t_{i+1}}-S_{t_i}$. 

We assume that the
producer aims to maximize the utility of profit at date $T$, that is, she solves the following control problem: 
\begin{equation}\label{discretwindproblem}
    \max_{\phi := (\phi_{t_i})_{i = 1}^{N-1}} \mathbb E\left[u\left(f(m_T)S_T - \sum_{i= 0}^{N-1}
\phi_{t_i} \Delta S_{t_i}-K|f(m_T) - \phi_{t_{N-1}}| \right)\right],
\end{equation}
where $u$ is a utility function (concave, increasing and satisfying
certain regularity conditions),  and $\phi := (\phi_{t_i})_{i = 1}^{N-1}$ belongs to a certain class
of admissible strategies. In particular, the process $\phi$ must be
adapted with respect to the filtration of the agent, generated by the
history of the process $S$, the history of the forecast process $m$ and a measure of forecast uncertainty if it is stochastic.
For the numerical resolution we assume that the agent has an exponential CARA utility function given by: 
$$ u(x) = 1-e^{-\alpha x},\quad \alpha >0.$$

To assess the importance of modeling the dynamics of forecast uncertainty, in the next section we perform the following numerical experiment. 
\begin{itemize}
    \item We consider two models, model A, which describes the dynamic evolution of forecast uncertainty, and model B, which does not include such a description. Model A, detailed in section \ref{sec:pos}, uses the log-inverse Gaussian predictive distribution and the forecast evolution given by equations (\ref{pos1}--\ref{pos2}). Model B is a simplified version of model A, with a constant diffusion coefficient of the forecast process:
    \begin{align}\label{eq:constantVol}
\frac{dm_t}{m_t} &= \sigma_m dW_t,
\end{align} 
where $\sigma_m$ is a constant such that $\sigma_m = V_0$. 
Since empirical studies show a negative correlation between the market price and the wind production forecasts \cite{kiesel2017econometric, feron2020price}, we assume $\langle W, B \rangle_t = \lambda t, \; \lambda <0, \; \forall t \in [0, T]$.
\item For each model, we compute the optimal feedback strategies $\phi^A_{t_i}(S,m,V)$ and $\phi^B_{t_i}(S,m)$, for $i=0,\dots,N-1$ by solving the problem \eqref{discretwindproblem} using the  Least Squares Monte Carlo algorithm. 
\item We then simulate the prices using model A, and compute the profit of the producer with the feedback strategies $\phi^A$ and $\phi^B$. The difference between the two profit amounts allows to quantify the loss from using model B, that is, from not taking into account the dynamic evolution of the forecast uncertainty, when the data follow model A. 
\end{itemize}

\subsection{Numerical resolution}

In the first part of this section we present the Least Squares Monte Carlo algorithm used to solve the control problem. Next we detail the parameter values and finally compare the profits  obtained in the case of model A and model B. 

\paragraph{Least Square Monte Carlo algorithm}

We consider the equivalent problem 

\begin{equation}\label{equidiscretewindproblem}
    \min_\phi \quad \mathbb E\left[\exp\left\{-\alpha\left(f(m_T)S_T - \sum_{i= 0}^{N-1}
\phi_{t_i} \Delta S_{t_i}-K|f(m_T) - \phi_{t_{N-1}}|\right)\right\}\right],
\end{equation}
and define its value function at each time step {$t_i, \; i = 0, \dots, N-1$,}
\begin{align*}
    & v_{t_i}(X) = \min_{\phi_{t_i}, \dots, \phi_{t_{N-1}}} \mathbb E_{t_i}^{X}\left[\exp\left\{-\alpha\left(f(m_T)S_T - \sum_{k= i}^{N-1}
\phi_{t_k} \Delta S_{t_k}-K|f(m_T) - \phi_{t_{N-1}}|\right)\right\}\right]
\end{align*}
where $X = (S,m, V)$ for Model A  and $X = (S, m)$
for Model B. 

Exploiting the exponential structure of the utility function, the dynamic programming principle takes the following form. 
\begin{align*}
%& v_{t_N}(S, m, V) = \exp\{-\alpha f(m)S \},\\
& v_{t_{N-1}}(X) = {\min_{\phi_{t_{N-1}}} \quad \mathbb E^{X}_{t_{N-1}} \left[ \exp\{\alpha( \phi_{t_{N-1}}\Delta S_{t_{N-1}}-f(m_{T_N})S_{T_N}+K|f(m_T) - \phi_{t_{N-1}}|)\}\right],}\\
& v_{t_i}(X) = \min_{\phi_{t_i}} \quad \mathbb E^{X}_{t_i} \left[ \exp\{\alpha \phi_{t_i}\Delta S_{t_i}\}\right],\; i = 0, \dots, N-2.
\end{align*}
%\st{ where  $\mathbb F:= (\mathcal{F}_t)_{t\le T}$ is the filtration generated by the processes $(S_t,m_t, V_t)_{t\in [0,T]}$ (resp $(S_t,m_t)_{t \in [0,T]}$). }\todo{We don't use the filtration in the text}\\
%\todo{Please write the value function as function of variables $(S,\mu,V)$ rather than stochastic processes}

For the numerical computation of the value functions, we use a regression approach based on adaptative local basis functions, described in Bouchard and Warin \cite{bouchard2012monte} and implemented in the open source library StOpt (The STochastic OPTimization library, see  \cite{gevret2018stochastic} for a detailed documentation). We briefly describe the algorithm below and refer the reader to \cite{bouchard2012monte, gevret2018stochastic} for further details. 

At each time step $t_i, \; i = 1,\dots, N$, the state space is partitioned into $Q$ cells, denoted by $D^{t_i}_{q},$ $q\in Q$, and on each cell, a linear local basis function $\psi_q$ is defined. We denote by $\beta_q\in \mathbb R^{d+1}$ the coefficients of the funciton $\psi_q$, where $d=3$ (resp.~d=2) is the dimension of the problem.

%At each time step $t_i, \; i = 1,\dots, N$, this approach consists in two stages. First, define a partition of $Q$ cells of the state space at time $t_i$ in which all cells  approximately contains the same number of particles. Second,  proceed to a linear regression of an $\mathcal{F}_{t_{i+1}}$- adapted quantity on each cell over linear or constant functions. That is, on each hypercube  of the state space partition  at time $t_i$, denoted $D^{t_i}_{q},$ $q\in Q$,  we define the local  function basis $\psi_q$, $q = 1, \dots, Q$. The linear or constant nature of the basis has to be set according to the performance of the regression in both cases.

%This approach consists, at each time step $t_i, \; i = 1,\dots, N$, in regressing an $\mathcal{F}_{t_{i+1}}$- adapted quantity over $K$ basis functions $\psi_1, \dots, \psi_K$ of the state space variables at time $t_i$. In our case, we choose Hermite polynomials of degree $d = 5$ for the basis functions. Hermite polynomials are a standard choice in this type of problems and the choice of the  degree is a trade-off between the quality of the estimation and avoiding over fitting.
%First, define a partition on the state space in which all cells  approximately contains the same number of particles. We then  proceed to a linear regression on the state space
%each cell over a  $d+1$ degree of freedom linear functions. That is, on each hypercube  $D^{t_i}_{q}$ $q\in Q$ cell of the partition of the state space at time $t_i$, we define the local linear function basis $\psi_q$, $q = 1, \dots, Q$.

Let $(X^j_{t_i})^{j = 1, \dots M}_{i = 1, \dots, N}$ be the Monte Carlo simulations of the discretized version of the processes $X:= (S,m,V)$ for model A (resp $X:= (S,m)$ for model B). We call these simulations the \emph{learning set}. Let $\phi^\ell_{t_i},\ell= 1,\dots,L$ be the discretized values of the control at time $t_i$. The algorithm for computing the optimal strategies consists in the following steps, performed backward in time, starting from $i=N-1$.

%We note $\widehat{\mathbb E}^M_i(x,\phi)$ an estimator of $\mathbb E\left[\hat v_{t_{i+1}}\times \exp\{\alpha \phi\Delta S_{t_i}\}| X_{t_i} = x\right]$ with $M$ Monte Carlo simulations. 
%We use a grid of size $L$ for the control $\phi_{t_i}, \;i= 0, \dots,N-1$ such that we evaluate the conditional expectation at the points $\phi^\ell_{t_i}, \; \ell = 1, \dots, L$ to estimate the value function $\hat v_{t_i}$ at time $t_i$. We begin backward, assume without loss of generality that the local basis are linear functions and use the following procedure at time $t_i$.

\begin{itemize}
   \item[1.] For each point $\phi^\ell_{t_i}$, we determine the vector $\hat \beta_q (\phi^l_{t_i})$ as follows: 
\begin{align*}
& \hat \beta^{N-1}_q(\phi^\ell_{t_{N-1}}) = \text{argmin}_{\beta_q \in \mathbb R^{d+1}}\\ &\qquad \sum_{x \in D^{t_{N-1}}_q} \left[  \psi_q(\beta_q,x) - \exp\{\alpha( \phi_{t_{N-1}}\Delta S_{t_{N-1}}-f(m_{T_N})S_{T_N}+K|f(m_T) - \phi_{t_{N-1}}|)\}\right]^2, \\
    & \hat \beta^i_q(\phi^\ell_{t_i}) = \text{argmin}_{\beta_q \in \mathbb R^{d+1}} \sum_{x \in D^{t_i}_q} \left[  \psi_q(\beta_q, x) - \hat v_{t_{i+1}}\times \exp\{\alpha \phi^\ell_{t_i}\Delta S_{t_i}\}\right]^2, \; i = 0, \dots, N-2
\end{align*}

    \item[2.] We then define the regression estimators for the conditional expectations: 
   \begin{align}\label{eq: VFreg}
       \widehat{e}_i^M(x,\phi^\ell_{t_i}) = \sum_{q = 1}^Q \psi_q(\beta_q(\phi^\ell_{t_i}), x) \mathbbm{1}_{x \in D_q}, \; i = 0, \dots, N-1
   \end{align} 
\item[3.]
The optimal feedback strategy and the value function are estimated as follows. $$\hat\phi^*_{t_i}(x) = \text{argmin}_{\phi_{t_i}^\ell, \, \ell = 1, \dots, L} \widehat{e}_i^M(x,\phi^\ell_{t_i}), \quad \hat v_{t_i}(x) = \widehat{e}_i^M(x,\hat\phi^*_{t_i}(x)) $$
\end{itemize}

%{\color{magenta} 

%For each point $\phi^\ell_{t_j}$, we give an estimator $C^M_{t_j}$ thanks to regression functions $(\Psi_k)_{k = 1, \dots, K}$ and determine the vector $\hat \alpha (\phi^\ell_{t_j})$ such that: 
    
%   $$\hat \alpha (\phi^\ell_{t_j}) = \text{argmin}_{\alpha \in \mathbb R^K} \sum_{i = 1}^M \left[ \sum_{k = 1}^K \Psi_k(X^i_{t_j})\alpha_k - v_{t_{j+1}}(X^i_{t_{j+1}}(\phi))\times \exp\{\alpha \phi^\ell_{t_j}\Delta S_{t}\}\right]^2$$
%    and we obtain: 
%   $$ C_j^M(x,\phi^\ell_{t_j}) = \sum_{k = 1}^K \psi_k(x)\hat \alpha (\phi^\ell_{t_j})_k $$
%}

To compare the gains in model A and in model B we then simulate $M'$ new trajectories of model A (the \emph{testing set}) and compute the gains using the optimal feedback strategies derived above.
%\begin{itemize}
%    \item[i.] We estimate the conditional expectation at each time step following the model of equation  \eqref{eq: VFreg}, using $M$ Monte Carlo simulations of the processes.
%    We estimate it when additional information is available about the evolution of the forecast error (Producer 1), thanks to the stochastic volatility. In this case, we denote the conditional expectation $\hat C^{M, 1}:=\left(\hat C^{M, 1}_i(.,.)\right)_{i = 1, \dots, N}$. And, we also estimate it in the case of the point forecasts updates (Producer 2) and denote it $\widehat{\mathbb E}^{M, 2} := \left(\widehat{\mathbb E}^{M, 2}_i(.,.)\right)_{i = 1, \dots, N}$. In order to have comparable results in the following steps, we calibrate the regression basis so that we obtain comparable level of gains on in sample data for both producers. This way, the differences observed in the gains (if any), in the next paragraphs, will only be triggered by the use, or not, of an additional information on the uncertainty variations.
%    \item[ii.] In a second step, we then simulate again $M'$ trajectories of the processes with the stochastic wind volatility. These simulations are supposed to represent the 'true' data delivered to the agents where the latter can, or cannot observe the information carried by the stochastic volatility. Using these trajectories and the conditional expectations $\hat C^{M, 1}$ and $\hat C^{M, 2}$ estimated in the first step, we choose in each case what is the appropriate strategy.
% \end{itemize}
In the next paragraph we detail the choice of parameters.

\paragraph{Setting of the numerical illustrations}
We use the empirical data and the parameters estimated in Section \ref{calibration} for the simulations of the present illustration. 

%Hence, we use the average value of the wind speed forecasts at lead times $36h$, $24h$ or $12h$, depending on the delivery hour we consider to set the initial value $m_0$ of the forecast process. L
We consider the delivery hour $12$ PM (noon),
and assume that forecast updates become available and the trading takes place at 12 PM on the previous day (this corresponds to the day-ahead trade), at 6PM on the previous day, at midnight, at 6AM on the delivery date, and at 12PM on the delivery date. Letting $t=0$ correspond to 12PM of the day preceding the delivery day, we then have: $N=4$ and $(t_0,t_1,t_2,t_3,t_4) = (0H, 6H, 12H, 18H, 24H)$.
The parameters of Model A are estimated as explained in Section \ref{methodcali}, and for model B we fix $m_0$ to the same value as in Model A, and $\sigma_m \rho\sqrt{V_0}$. The value of price volatility is calibrated as explained in \cite{feron2020price}. The drift $\mu_S$, is not fixed for now, we will make it vary in the next paragraph for our study.
The absolute risk aversion coefficient $\alpha$ was chosen in an ad hoc manner, but in such a way that its numerical value is compatible with the average value of producer's revenues. 
The values of all parameters are summarized in Table \ref{parameters}. 

\begin{table}[ht]
\centering
   \begin{tabular}{ |c|c||c|c|}
     \hline
     Parameter & Value & Parameter & Value  \\ \hline
     $S_0$ & $40$ \euro{}/MWh   &  $V_0$ &  0.032  \\ \hline
     $\sigma^{S}$ & $6$ \euro{}/MWh.h$^{1/2}$ &$\rho$ & 0.16  \\ \hline
     $m_0$ & 5.38 m/s &
       $b$ & 0.035  \\
       \hline
      $\lambda$ & -0.08 & $\alpha$ & 0.01 \euro{}$^{-1}$ \\ \hline
    $m_{\min}$ & 3.3 m/s & $m_{\max}$ & 25 m/s \\ \hline
     $K$ & 10 \euro{}/MWh &  &\\ \hline%$\kappa$ & 1 \euro{}/MWh \\ \hline
   \end{tabular}
   \caption{Parameters of the model}
   \label{parameters}
 \end{table}

For the estimation of the conditional expectations (the training set) we use $M = 200000$ MC trajectories, and the grid size $Q=15\times 15$ for the two-dimensional model and $Q = 15\times 15\times 15$ for the three-dimensional model. The actual grids are determined by the algorithm in an adaptive manner. \\
%In addition, by comparing model A and model B we compare a three dimensional problem with a two dimensional problem. For this reason, setting a comparable partition of the space when using local basis regression is not an easy task. We decided to set the grid such that the residual term of the regression that we denote $\varepsilon_A$ and $\varepsilon_B$, and their variance, are almost the same -- relative differences $ \frac{|\varepsilon_A-\varepsilon_B|}{\varepsilon_B}< 1\%$, $ \frac{|\varepsilon^2_A-\varepsilon^2_B|}{\varepsilon^2_B}< 1\%$; which allows to think that the quality of the regression in model A and model B is comparable. 

To evaluate the gains of the two strategies (the test set), we use $M'= 1 000 000$ trajectories. 
The control values (position of the agent) are also discretized on a grid, which depends on the setting of the problem. In the numerical experiments presented in the next section, we consider three different settings. When the price is martingale, we allow for positions between $-1$ and $1$ with a step size of $0.01$. When the price has a positive or negative trend, we allow for positions between $-5$ and $5$ with a step size size $0.05$.  
These bounds in accordance with the observed shapes of the strategies during experiments, and by making a trade-off between accuracy and the computational cost.

 \paragraph{Results}
 After deriving the optimal strategies  $(\phi^{j*}_{t_i})_{0\le i \le N-1}$ for model A and model B, we then simulate $M'$ new trajectories under model A and compute the realized profit for the two producers:
 $$f(m^j_T)S^j_T - \sum_{i= 0}^{N-1}
\phi^{j*}_{t_i} \Delta S^j_{t_i}.$$
This computation is repeated multiple times to evaluate the sensitivity of the average profit to the presence of price trend and the wind forecast uncertainty level. %he initial volatility parameter calibrated with the data ($\sigma_m$ = 0.029 m/s.s$^{1/2}$, or equivalently $V_0$ = 0.032 m/s.s$^{1/2} $) produces a wind speed between  $15$ km/h and $45$ km/h in our simulations. But this is the case because we are in a plain (see Figure \ref{locations}) and it might have greater ranges in other landscapes where eolian are present such as sea coasts. This is why we study greater volatility levels. 
Figure \ref{Fig:windspeed} shows the range of wind speeds for different wind forecast uncertainty levels. 

%\todo{Dans Figure 14, il faut remplacer volatility par uncertainty et in function par as function}

\begin{figure}[ht]
      \centering
      \caption{Wind speed range as function of the uncertainty parameter.}
      \includegraphics[width = 0.6\textwidth]{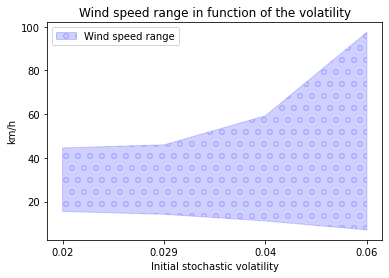}
      \label{Fig:windspeed}
\end{figure}

%We chose the volatility values so that they produce acceptable ranges of wind speed for wind energy production (from 10 km/h up to 90 km/h).
%An important observation is that setting the volatility as expressed in \eqref{eq:fixsig} produces the same ranges of wind speed in the constant volatility and stochastic volatility settings. It makes the profits in Figure \ref{Fig:profits} comparable but more importantly, it confirms that the difference we observe in the profits is only due to the variation of the uncertainty on a side and its non-variation on the other side\todo{est ce que je laisse la vol a 0.15 meme si ça donne des vitesses extremes pas du tout realistes ?}.\\
%\todo{The volatility is way too high, and one does not need to simulate anything with the constant vol model - see above}

%\todo{refaire mus=0 avec 500000 trajectories}

\begin{figure}[ht]
    \centering
      \caption{Realized profits for different values of the price trend }\label{Fig:profits}
      %\includegraphics[width = 0.48\textwidth]{pic/profit0.png}
      %\vfill
      \includegraphics[width = 0.48\textwidth]{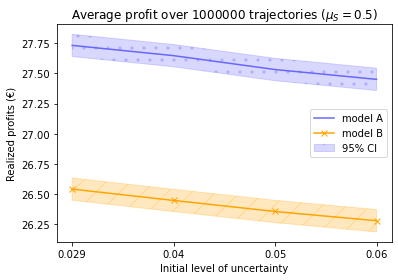}
      \includegraphics[width = 0.48\textwidth]{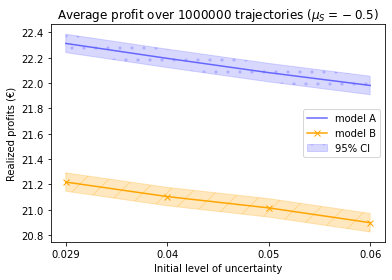}
      \includegraphics[width = 0.48\textwidth]{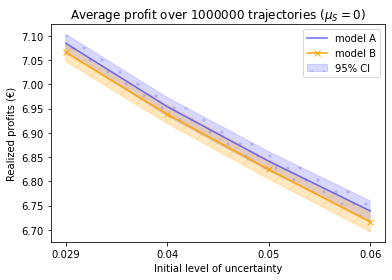}
\end{figure}

 { The simulated average profit for the two producers is shown in Figure \ref{Fig:profits}.  We see that in both models, the average profit decreases as function of the forecast uncertainty level. The presence of a price trend leads both for higher profits for the two models and a clear improvement of performance of model A compared to model B. 
In the case where the price is a martingale the result is not obvious since the confidence intervals of the average profits overlap to a large extent: this result  needs to be further investigated. 
 %However depending on the trend one model may be preferred to the other. For a martingale price and small trends such as $\mu_S = 1, -1 $ \euro{}/MWh.h the control strategy taking into account the evolution of the uncertainty performs better than the point forecast one. 
 The observations made from Figure \ref{Fig:profits} are  confirmed in Table \ref{tab:profits} where we show the relative profits computed as follows: 
 $$\frac{\text{Profit(Model A)}-\text{Profit(Model B)}}{\text{Profits(Model B)}}\times 100.$$

\begin{table}[ht]
    \centering
\begin{tabular}{|c|c|c|c|}
\hline
$\mu_S$ (\euro{}/MWh.h) &  Relative profits \\
& in $\%$\\
\hline
$0$        &      $+0.27$ \\
$0.5$       &     {$\mathbf{+4.49}$}  \\
$-0.5$       &      {$\mathbf{+ 5.15}$} \\
%2        &      + 8.2 \\
%-2 &  + 9.6\\
\hline
\end{tabular}
    \caption{Relative profits of Model A (taking into account stochastic forecast uncertainty) compared to Model B. Bold values are statistically significant at 95\% level. }
    \label{tab:profits}
\end{table}

 %Nevertheless when the trend in the intraday price becomes greater such as $\mu_S = 2, -2 $ \euro{}/MWh.h then it is preferable to consider only the constant volatility. The more 'certain' is the price, the less it worthies taking into account the forecasts of the terminal production since high profits can be done only by trading. 
 In the literature, estimations of the trend in  the intraday electricity market price are scarce. A recent empirical study by Glas et Al. \cite{glas2020intraday} reports the presence of a small trend composed of a constant part (0.0433 \euro{}/MWh$^2$) and a permanent price impact (0.0017 \euro{}/MWh$^2$). Hence taking into account the evolution of the forecast uncertainty as it is done in the model we propose seem to be adapted to the market reality and impacts the strategies in a way that increases significantly the profits.}

\section*{Appendix}
\subsection{Talagrand diagrams and PIT histograms for lead times 36h and 48h}\label{HistoAppend}
\begin{figure}[ht]
      \centering
      \caption{Talagrand diagrams for the wind speed and the temperature, lead time $36h00$}
      \includegraphics[width = 0.48\textwidth]{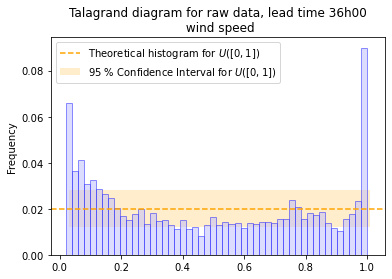}
      \includegraphics[width = 0.48\textwidth]{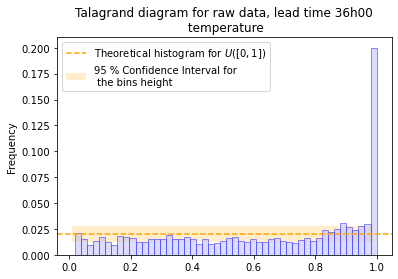}
   %\hfill
    %\includegraphics[width =0.48\textwidth]{pic/real.png}
\end{figure}

\begin{figure}[ht]
      \centering
      \caption{PIT histogram for the wind speed and the temperature, lead time $36h00$ }
      \includegraphics[width = 0.48\textwidth]{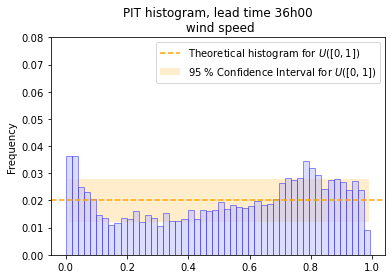}
      \includegraphics[width = 0.48\textwidth]{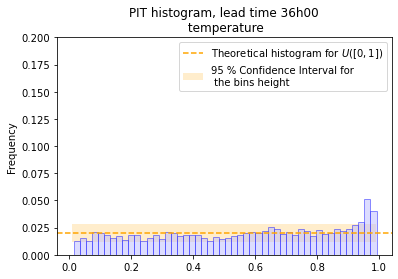}
   %\hfill
    %\includegraphics[width =0.48\textwidth]{pic/real.png}
\end{figure}
\newpage

\begin{figure}[ht]
      \centering
      \caption{Talagrand diagrams for the wind speed and the temperature, lead time $48h00$}
      \includegraphics[width = 0.48\textwidth]{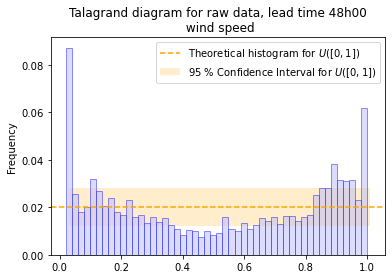}
      \includegraphics[width = 0.48\textwidth]{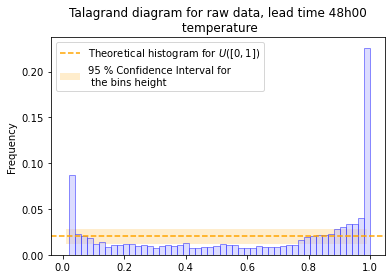}
   %\hfill
    %\includegraphics[width =0.48\textwidth]{pic/real.png}
\end{figure}

\begin{figure}[ht]
      \centering
      \caption{PIT histogram for the wind speed and the temperature, lead time $48h00$ }
      \includegraphics[width = 0.48\textwidth]{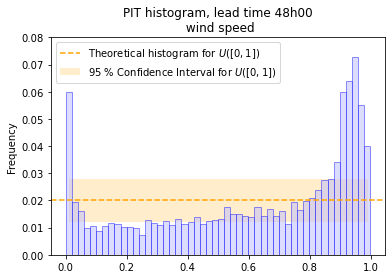}
      \includegraphics[width = 0.48\textwidth]{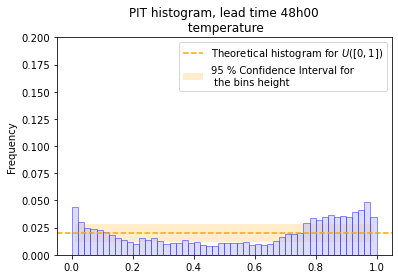}
   %\hfill
    %\includegraphics[width =0.48\textwidth]{pic/real.png}
\end{figure}

\newpage

%\bibliographystyle{siam}
%\bibliography{bibliography.bib}

\begin{thebibliography}{10}

\bibitem{aas2006generalized}
{\sc K.~Aas and I.~H. Haff}, {\em The generalized hyperbolic skew
  student’st-distribution}, Journal of financial econometrics, 4 (2006),
  pp.~275--309.

\bibitem{aid2016optimal}
{\sc R.~A{\"\i}d, P.~Gruet, and H.~Pham}, {\em An optimal trading problem in
  intraday electricity markets}, Mathematics and Financial Economics, 10
  (2016), pp.~49--85.

\bibitem{badosa2017day}
{\sc J.~Badosa, E.~Gobet, M.~Grangereau, and D.~Kim}, {\em Day-ahead
  probabilistic forecast of solar irradiance: a stochastic differential
  equation approach}, in Forecasting and Risk Management for Renewable Energy,
  Springer, 2017, pp.~73--93.

\bibitem{baran2014probabilistic}
{\sc S.~Baran}, {\em Probabilistic wind speed forecasting using bayesian model
  averaging with truncated normal components}, Computational Statistics \& Data
  Analysis, 75 (2014), pp.~227--238.

\bibitem{baran2015log}
{\sc S.~Baran and S.~Lerch}, {\em Log-normal distribution based ensemble model
  output statistics models for probabilistic wind-speed forecasting}, Quarterly
  Journal of the Royal Meteorological Society, 141 (2015), pp.~2289--2299.

\bibitem{baran2016mixture}
\leavevmode\vrule height 2pt depth -1.6pt width 23pt, {\em Mixture emos model
  for calibrating ensemble forecasts of wind speed}, Environmetrics, 27 (2016),
  pp.~116--130.

\bibitem{barndorff1997processes}
{\sc O.~E. Barndorff-Nielsen}, {\em Processes of normal inverse gaussian type},
  Finance and stochastics, 2 (1997), pp.~41--68.

\bibitem{belomestny2010regression}
{\sc D.~Belomestny, A.~Kolodko, and J.~Schoenmakers}, {\em Regression methods
  for stochastic control problems and their convergence analysis}, SIAM Journal
  on Control and Optimization, 48 (2010), pp.~3562--3588.

\bibitem{bensoussan2016cox}
{\sc A.~Bensoussan and A.~Brouste}, {\em Cox--ingersoll--ross model for wind
  speed modeling and forecasting}, Wind Energy, 19 (2016), pp.~1355--1365.

\bibitem{bouchard2012monte}
{\sc B.~Bouchard and X.~Warin}, {\em Monte-carlo valuation of american options:
  facts and new algorithms to improve existing methods}, in Numerical methods
  in finance, Springer, 2012, pp.~215--255.

\bibitem{collet2017optimal}
{\sc J.~Collet, O.~F{\'e}ron, and P.~Tankov}, {\em Optimal management of a wind
  power plant with storage capacity}, in Forecasting and Risk Management for
  Renewable Energy, Springer, 2017, pp.~229--246.

\bibitem{dufresne1990distribution}
{\sc D.~Dufresne}, {\em The distribution of a perpetuity, with applications to
  risk theory and pension funding}, Scandinavian Actuarial Journal, 1990
  (1990), pp.~39--79.

\bibitem{feron2020price}
{\sc O.~F{\'e}ron, P.~Tankov, and L.~Tinsi}, {\em Price formation and optimal
  trading in intraday electricity markets}, arXiv preprint arXiv:2009.04786,
  (2020).

\bibitem{gevret2018stochastic}
{\sc H.~Gevret, N.~Langren{'e}, J.~Lelong, X.~Warin, and A.~Maheshwari}, {\em
  Stochastic optimization library in c++}, hal-01361291v1,  (2018).

\bibitem{glas2020intraday}
{\sc S.~Glas, R.~Kiesel, S.~Kolkmann, M.~Kremer, N.~G. von Luckner,
  L.~Ostmeier, K.~Urban, and C.~Weber}, {\em Intraday renewable electricity
  trading: Advanced modeling and numerical optimal control}, Journal of
  Mathematics in Industry, 10 (2020), p.~3.

\bibitem{gneiting2007probabilistic}
{\sc T.~Gneiting, F.~Balabdaoui, and A.~E. Raftery}, {\em Probabilistic
  forecasts, calibration and sharpness}, Journal of the Royal Statistical
  Society: Series B (Statistical Methodology), 69 (2007), pp.~243--268.

\bibitem{gneiting2005calibrated}
{\sc T.~Gneiting, A.~E. Raftery, A.~H. Westveld~III, and T.~Goldman}, {\em
  Calibrated probabilistic forecasting using ensemble model output statistics
  and minimum crps estimation}, Monthly Weather Review, 133 (2005),
  pp.~1098--1118.

\bibitem{iversen2016short}
{\sc E.~B. Iversen, J.~M. Morales, J.~K. M{\o}ller, and H.~Madsen}, {\em
  Short-term probabilistic forecasting of wind speed using stochastic
  differential equations}, International Journal of Forecasting, 32 (2016),
  pp.~981--990.

\bibitem{jeanblanc2009mathematical}
{\sc M.~Jeanblanc, M.~Yor, and M.~Chesney}, {\em Mathematical methods for
  financial markets}, Springer Science \& Business Media, 2009.

\bibitem{kiesel2017econometric}
{\sc R.~Kiesel and F.~Paraschiv}, {\em Econometric analysis of 15-minute
  intraday electricity prices}, Energy Economics, 64 (2017), pp.~77--90.

\bibitem{lerch2013comparison}
{\sc S.~Lerch and T.~L. Thorarinsdottir}, {\em Comparison of non-homogeneous
  regression models for probabilistic wind speed forecasting}, Tellus A:
  Dynamic Meteorology and Oceanography, 65 (2013), p.~21206.

\bibitem{longstaff2001valuing}
{\sc F.~A. Longstaff and E.~S. Schwartz}, {\em Valuing american options by
  simulation: a simple least-squares approach}, The review of financial
  studies, 14 (2001), pp.~113--147.

\bibitem{pinson2007trading}
{\sc P.~Pinson, C.~Chevallier, and G.~N. Kariniotakis}, {\em Trading wind
  generation from short-term probabilistic forecasts of wind power}, IEEE
  Transactions on Power Systems, 22 (2007), pp.~1148--1156.

\bibitem{pinson2013wind}
{\sc P.~Pinson et~al.}, {\em Wind energy: Forecasting challenges for its
  operational management}, Statistical Science, 28 (2013), pp.~564--585.

\bibitem{raftery2005using}
{\sc A.~E. Raftery, T.~Gneiting, F.~Balabdaoui, and M.~Polakowski}, {\em Using
  bayesian model averaging to calibrate forecast ensembles}, Monthly weather
  review, 133 (2005), pp.~1155--1174.

\bibitem{skajaa2015intraday}
{\sc A.~Skajaa, K.~Edlund, and J.~M. Morales}, {\em Intraday trading of wind
  energy}, IEEE Transactions on power systems, 30 (2015), pp.~3181--3189.

\bibitem{tan2018optimal}
{\sc Z.~Tan and P.~Tankov}, {\em Optimal trading policies for wind energy
  producer}, SIAM Journal on Financial Mathematics, 9 (2018), pp.~315--346.

\bibitem{thorarinsdottir2010probabilistic}
{\sc T.~L. Thorarinsdottir and T.~Gneiting}, {\em Probabilistic forecasts of
  wind speed: Ensemble model output statistics by using heteroscedastic
  censored regression}, Journal of the Royal Statistical Society: Series A
  (Statistics in Society), 173 (2010), pp.~371--388.

\bibitem{tsitsiklis1999optimal}
{\sc J.~N. Tsitsiklis and B.~Van~Roy}, {\em Optimal stopping of markov
  processes: Hilbert space theory, approximation algorithms, and an application
  to pricing high-dimensional financial derivatives}, IEEE Transactions on
  Automatic Control, 44 (1999), pp.~1840--1851.

\bibitem{wilks2002smoothing}
{\sc D.~S. Wilks}, {\em Smoothing forecast ensembles with fitted probability
  distributions}, Quarterly Journal of the Royal Meteorological Society: A
  journal of the atmospheric sciences, applied meteorology and physical
  oceanography, 128 (2002), pp.~2821--2836.

\bibitem{zugno2013trading}
{\sc M.~Zugno, T.~J{\'o}nsson, and P.~Pinson}, {\em Trading wind energy on the
  basis of probabilistic forecasts both of wind generation and of market
  quantities}, Wind Energy, 16 (2013), pp.~909--926.

\end{thebibliography}
\end{document}